\documentclass[11pt,letterpaper]{article}

\usepackage[scale=0.8,centering]{geometry}
\usepackage{graphicx}
\usepackage{amsmath}
\usepackage{amsthm,amssymb,amsfonts,mathtools,mathrsfs,bm}
\usepackage{epsfig,color,cases,url}
\usepackage[dvips]{hyperref}
\usepackage[noadjust]{cite}
\usepackage[caption=false,font=footnotesize]{subfig}
\usepackage{booktabs,multirow}
\usepackage{authblk}

\usepackage{enumitem}
\newcommand{\Natural}{\mathbb{N}}
\newcommand{\Naturalstar}{\mathbb{N}_*}
\newcommand{\Real}{\mathbb{R}}

\newcommand{\innprod}[2]{\left\langle{#1},{#2}\right\rangle}

\newcommand{\norm}[1]{\left\|{#1}\right\|}

\newcommand{\lev}[1]{\ensuremath{\operatorname{lev_{\leq #1}}}}

\DeclareMathOperator{\argmin}{arg\,min}

\DeclareMathOperator{\sign}{sgn}

\DeclareMathOperator{\Fix}{Fix}

\DeclareMathOperator{\supp}{supp}

\DeclareMathOperator{\expect}{\mathsf{E}}
\DeclareMathOperator{\prob}{Prob}

\DeclareMathOperator{\shr}{shr}

\theoremstyle{definition}
\newtheorem{theorem}{Theorem}
\newtheorem{algo}{Algorithm}

\newtheorem{definition}{Definition}

\newtheorem{assumption}{Assumption}

\begin{document}

\title{Generalized Thresholding\\ and Online Sparsity-Aware Learning in a
  Union of Subspaces}

\author{Konstantinos Slavakis,\authorcr University of Minnesota, Digital
  Technology Center (DTC),\authorcr 487 Walter Library, 117 
  Pleasant St.\ SE,\authorcr Minneapolis, MN 55455, USA.\authorcr Email:
  \url{slavakis@dtc.umn.edu} \and Yannis Kopsinis,\authorcr University of
  Granada, Dept.\ of Applied Physics, Granada,
  Spain.\authorcr \url{ykopsinis@gmail.com} \and Sergios
  Theodoridis,\authorcr University of Athens, Dept.\ of
  Informatics \& Telecommunications,\authorcr Athens, Greece.\authorcr
  Email: \url{stheodor@di.uoa.gr} \and Stephen 
  McLaughlin,\authorcr Heriot Watt University, School of
  Engineering \& Physical Sciences,\authorcr Edinburgh, UK.\authorcr Email:
  \url{S.McLaughlin@hw.ac.uk}}

\date{}

\maketitle

\begin{abstract}
This paper studies a sparse signal recovery task in time-varying
(time-adaptive) environments. The contribution of the paper to
sparsity-aware online learning is threefold; first, a Generalized
Thresholding (GT) operator, which relates to both convex and non-convex
penalty functions, is introduced. This operator embodies, in a unified way,
the majority of well-known thresholding rules which promote
sparsity. Second, a non-convexly constrained, sparsity-promoting, online
learning scheme, namely the Adaptive Projection-based Generalized
Thresholding (APGT), is developed that incorporates the GT operator with a
computational complexity that scales linearly to the number of
unknowns. Third, the novel family of partially quasi-nonexpansive mappings
is introduced as a functional analytic tool for treating the GT
operator. By building upon the rich fixed point theory, the previous class
of mappings helps us, also, to establish a link between the GT operator and
a union of linear subspaces; a non-convex object which lies at the heart of
any sparsity promoting technique, batch or online. Based on such a
functional analytic framework, a convergence analysis of the APGT is
provided. Furthermore, extensive experiments suggest that the APGT exhibits
competitive performance when compared to computationally more demanding
alternatives, such as the sparsity-promoting Affine Projection Algorithm
(APA)- and Recursive Least Squares (RLS)-based techniques. 
\end{abstract}

\section{Introduction}\label{sec:intro}

Sparsity-aware learning has been a topic at the forefront of research over
the last ten years or so \cite{CandesRombergTao06,
  Donoho2006}. Considerable effort has been invested in developing
efficient schemes for the recovery of sparse signal/parameter
vectors. However, most of these efforts have focussed on batch processing,
via the \textit{Compressed Sensing} or \textit{Sampling (CS)} framework.
In CS, an iterative algorithm is mobilized to solve the estimation task
once \textit{all} measurements (training data) have been collected by the
processing unit \cite{CandesRombergTao06, Donoho2006, BlumensathJST2010,
  Blumensath.IHT, Foucart.IHT}. It is only very recently that online
(time-adaptive) algorithms have been developed, where the training data are
processed sequentially, and the sparse signal to be recovered has the
freedom to be time-varying \cite{ChenHero09, Angelosantejournal2010,
  babadisparls, mileounisspadomp, myyy.soft.thres.icassp10,
  Kopsinis.Slavakis.TheodoridisIEEESP2011,
  Su_L0_LMS_performance_IEEESP2012}. Both CS and online techniques share a
common strategy, namely \textit{thresholding}; i.e, a thresholding rule is
used to impose sparsity-aware a-priori knowledge: some of the components of
the signal/vector to be estimated are kept intact, while the rest of them
are shrunk under some user-defined rule. Two thresholding operators
dominate the literature: (i) \textit{hard} thresholding, a brute force
method, where shrinking is achieved by setting the size of some of the
vector components to zero, and (ii) \textit{soft} thresholding, where the
shrinking operation is based on the (weighted) $\ell_1$-norm of the vector.

A large number of thresholding operators have been studied thoroughly, both
in theoretical and experimental contexts, mainly within the statistics
community \cite{Frank1993, Tibshirani.Lasso, Gao1997, Gao1998, Tao2000,
  Antoniadis2001, Fan_Li2001, Zou2006, Antoniadis2007, ZouLi2008, She.09,
  Zhang2010, Friedman.etal.coordinate.descent.10, SimonEtal2011,
  Mazumder2011}. It is by now well-established that hard thresholding, a
discontinuous operator, has a tendency for larger variance of the
estimates. Moreover, due to its discontinuity, hard thresholding can lead
to instabilities, in the sense of being sensitive to small changes in the
training data \cite{Antoniadis2007}. Soft-thresholding, is a continuous
operator, that tends to introduce bias in the estimates. Therefore,
alternative thresholding rules have been proposed in an effort to overcome
these drawbacks \cite{Frank1993, Antoniadis2001, Fan_Li2001, Zhang2010,
  Mazumder2011}. These advances in thresholding operators are strongly
connected to optimization tasks; they are obtained by minimizing squared
error terms regularized by, usually, \textit{non-convex} penalty functions.

 The contribution of this paper is threefold. First, the
 \textit{generalized thresholding} (GT) operator is introduced, which
 encompasses classical hard and soft thresholding rules, as well as the
 recent advances of \cite{Frank1993, Tibshirani.Lasso, Gao1997, Gao1998,
   Tao2000, Antoniadis2001, Fan_Li2001, Zou2006, Antoniadis2007, ZouLi2008,
   Zhang2010, Friedman.etal.coordinate.descent.10, SimonEtal2011,
   Mazumder2011}. Moreover, the proposed framework, motivated by the rich
 \textit{fixed point} theory \cite{GoebelKirk,
   bauschke.combettes.book}, is general enough to provide means for
 designing novel thresholding rules and/or incorporating a priori
 information associated with the sparsity level, i.e., the number of
 nonzero components, of the sparse vector to be recovered. More
 importantly, GT is also allowed to \textit{non-convexly} constrain the
 unknown vector.

Second, the GT operator is incorporated into a signal/parameter estimation
framework. Here, we choose the set theoretic estimation approach
\cite{CombettesFoundations}, and in particular its online version,
introduced in \cite{YamadaOguraAPSMNFAO} and extended in
\cite{KostasAPSMatNFAO, Sy.qne.siopt}. In particular, the \textit{Adaptive
  Projection-based Generalized Thresholding (APGT)} algorithm is proposed
having three important merits. a) It is an online algorithm, b) it promotes
sparse solutions effectively via the flexibility provided by the GT
operator and c) its computational complexity \textit{scales linearly to the
  number of unknowns}.  With respect to performance, although APGT shows a
low computational load, the experimental validation of
Section~\ref{sec:simulations} demonstrates that it exhibits a competitive
performance even when compared to very recently developed,
sparsity-promoting, and computationally more demanding alternatives, such
as the APA- and RLS-based techniques \cite{Angelosantejournal2010,
  Werner.Diniz, Hoshuyama_ipapa_2004, paleologu2010efficient}.

It should be noted that the adopted set theoretic estimation framework was
also utilized in \cite{Kopsinis.Slavakis.TheodoridisIEEESP2011}, where
sparsity was induced via $\ell_1$-based constraints, well-known to be
convex and intimately connected to soft thresholding operations. In
contrast, the fact that the GT operator is a ``non-convex'' mapping poses
certain challenges for the convergence analysis of the
algorithm. Specifically, the existing theory \cite{YamadaOguraAPSMNFAO,
  KostasAPSMatNFAO, Sy.qne.siopt} which, so far, has been developed
around convex sets and constraints is not rich enough to cover the APGT
case. In order to theoretically support the incorporation of GT into
learning mechanisms, such as the APGT, a novel family of operators,
hereafter referred to as \textit{partially quasi-nonexpansive mappings,} is
introduced, to the best of our knowledge, for the first time. It is the
introduction of the partially quasi-nonexpansive mappings and their nice
properties, which allowed the convergence analysis of APGT to be
developed. These operators serve as a sound theoretical tool which allows
the use of variational analysis \cite{Rockafellar.Wets} and fixed point
theory \cite{GoebelKirk, bauschke.combettes.book} to attack
\textit{non-convexly} constrained learning problems. It is shown that GT
belongs to this class of nonlinear mappings, with its \textit{fixed point
  set} being a union of subspaces; a non-convex object which lies at the
heart of any sparsity-promoting technique \cite{Lu.Do.2008,
  eldar.michali.union.2009}.

It should be stressed that, propelled by such a generic operator
theoretical framework, the proposed GT mapping offers a sound mathematical
basis for infusing sparsity arguments into both batch (CS) and online
approaches, beyond the set-theoretic framework adopted here. Moreover, the
present manuscript shows a value beyond sparsity-aware learning. Through
the novel concept of the partially quasi-nonexpansive mappings, this
study stands also as the first step toward the extension of
\cite{YamadaOguraAPSMNFAO, KostasAPSMatNFAO, Sy.qne.siopt} to
non-convexly constrained online learning tasks.

The remainder of the paper is organized as follows. The problem under
consideration is stated in Section \ref{sec:problem}. In Section
\ref{sec:GT}, the GT operator is introduced. The proposed
APGT algorithm is given in Section \ref{sec:algo}, together with its
properties and the definition of the novel family of partially
quasi-nonexpansive mappings. Section \ref{sec:simulations} contains the
experimental validation of APGT. A number of appendices support
theoretically the developments exposed throughout the paper. More
specifically, in App.~\ref{sec:properties.T_K} the properties of the
generalized thresholding operator are studied rigorously, and the
convergence analysis of the proposed algorithm is performed in
App~\ref{sec:analysis.algo}. A preliminary version of this study was
presented in \cite{gt.icassp.12}.

\section{Problem Statement and Related Work}\label{sec:problem}

We will denote the set of all non-negative integers, positive integers, and
real numbers by $\Natural$, $\Naturalstar$, and $\Real$,
respectively. Given any integers $j_1, j_2$, such that $j_1\leq j_2$, let
$\overline{j_1, j_2} \coloneqq \{j_1, j_1+1, \ldots,j_2\}$.

The stage for discussion will be the Euclidean space $\Real^L$, where
$L\in\Naturalstar$. Given any pair of vectors $\bm{a}_1, \bm{a}_2\in
\Real^L$, the inner product in $\Real^L$ is defined as the classical
vector-dot product $\innprod{\bm{a}_1}{\bm{a}_1} \coloneqq
\bm{a}_1^{\top} \bm{a}_2$, where $\top$ stands for vector/matrix
transposition. The induced norm is $\norm{\cdot}\coloneqq
\sqrt{\innprod{\cdot}{\cdot}}$.

Our discussion will revolve around the following celebrated linear model:
\begin{equation}
y_n = \bm{u}_n^{\top}\bm{a}_*  + v_n, \quad\forall n\in\Natural,
\label{eq:regression_Model2}
\end{equation}
where $\bm{a}_*\in \Real^L$ is an unknown vector/signal,
$(\bm{u}_n,y_n)_{n\in\Natural} \subset \Real^L \times \Real$ is a sequence
of known training data, and $(v_n)_{n\in\Natural}$ stands for the noise
process. In other words, the unknown $\bm{a}_*$ is ``sensed'' by a sequence
of input vectors $(\bm{u}_n)_{n\in\Natural}$, via the inner product of
$\Real^L$, in order to produce the noisy outputs
$(y_n)_{n\in\Natural}$. The vector $\bm{a}_*$ is considered to be
\textit{sparse}, i.e., most of its components are zero. If we define
$\norm{\bm{a}_*}_0$ to stand for the number of non-zero components of
$\bm{a}_*$, then the assumption that $\bm{a}_*$ is sparse can be
equivalently given by $K_* \coloneqq \norm{\bm{a}_*}_0\ll L$, and the
vector $\bm{a}_*$ will be called \textit{$K_*$-sparse.}

This study attacks the following inverse problem: estimate the unknown
sparse vector $\bm{a}_*$ by utilizing the sequence of training data
$(\bm{u}_n, y_n)_{n\in\Natural}$. A family of algorithms which shares a
similar objective is the \textit{Compressed Sensing} or \textit{Sampling
  (CS)} framework \cite{CandesRombergTao06, Donoho2006}. Given a fixed
number $N\in \Naturalstar$ of training data $(\bm{u}_i,
y_i)_{i=n-N+1}^{n}$, a CS algorithm is mobilized in order to compute an
estimate $\bm{a}_n$ of $\bm{a}_*$. CS belongs to the class of batch
algorithms, i.e., in the case where the datum $(\bm{u}_{n+1},y_{n+1})$
enters the system, a CS algorithm starts \textit{from scratch}, and
triggers a generally time consuming iterative procedure which operates on
the data $(\bm{u}_i, y_i)_{i=n-N+2}^{n+1}$ for computing the updated estimate
$\bm{a}_{n+1}$ of $\bm{a}_*$. In contrast to \textit{batch learning}
approaches, this manuscript focuses on sparsity-aware \textit{online
  learning}, i.e., an algorithmic framework which satisfies the following
requirements.

\begin{enumerate}[leftmargin=0pt,itemindent=20pt]
\item The estimates of $\bm{a}_*$ should be updated in a simple and
  efficient way every time that a new datum $(\bm{u}_n,y_n)$ enters
  the system.  The need to mobilize an optimization procedure from
  scratch, for every new datum $(\bm{u}_n,y_n)$, as in CS, should be
  avoided.
\item The operations needed in order to update the estimate should be
  of low computational complexity; hopefully of \textit{linear
    complexity with respect to the number of unknowns,} i.e.,
  $\mathcal{O}(L)$.
\item The unknown $\bm{a}_*$ has also the freedom to be
  \textit{time-varying}. Thus, an online learning scheme should be also
  able to \textit{quickly track} any variations of $\bm{a}_*$.
\end{enumerate}

The mainstream of sparsity-promoting online methods exploits training
data $(\bm{u}_n,y_n)_{n\in\Natural}$ in the context of classical adaptive
filtering \cite{Haykin}; a quadratic objective function is used to
quantify the designer's perception of loss. Additionally, a convex
differentiable function is regularized by a sparsity promoting term,
usually one that builds around the $\ell_1$ norm penalty function, and a
minimizer of the resulting optimization task is sought either in the RLS or
the LMS rationale, e.g., \cite{ChenHero09, Angelosantejournal2010,
  babadisparls, mileounisspadomp}. Another sparsity-promoting methodology,
where different components of the vector estimates are weighted under
several user-defined rules, is given by \textit{proportionate-type}
schemes \cite{Hoshuyama_ipapa_2004, paleologu2010efficient,
  Werner.Diniz}. Very recently, a novel online method for the recovery of
sparse signals, based on set theoretic estimation arguments
\cite{CombettesFoundations, tsy.sp.magazine}, was developed in
\cite{Kopsinis.Slavakis.TheodoridisIEEESP2011}, and extended for
distributed learning in \cite{Simos.sparse.distr}.

The set theoretic estimation philosophy departs from the standard approach
of constructing a loss function first; instead, it initially identifies a
\textit{set} of solutions which are in \textit{agreement} with the
available measurements as well as the available a-priori knowledge. A
popular strategy is to define, at each time instance $n\in\Natural$, a
closed convex subset of $\Real^L$, by means of the training data pair
$(\bm{u}_n,y_n)$, to contain the unknown $\bm{a}_*$ with high
probability. Different alternatives exist on how to ``construct'' such
convex regions. A popular choice takes the form of a \textit{hyperslab}
around $(\bm{u}_n,y_n)$, which is defined as:
\begin{equation}
S_n[\epsilon_n] \coloneqq \bigl \{\bm{a}\in\Real^L:\ \bigl|\bm{u}_n^{\top}
\bm{a} - y_n \bigr| \leq \epsilon_n \bigr\},\quad \forall n
\in\Natural,\label{eq:Hyperslab}
\end{equation}
for some user-defined tolerance $\epsilon_n\geq 0$, and for $\bm{u}_n\neq
\bm{0}$. The parameter $\epsilon_n$ determines, essentially, the width of
the hyperslabs, and it implicitly models the effects of the noise, as well
as various other uncertainties, like measurement inaccuracies, calibration
errors, etc. For example, if the noise were bounded, i.e., $\exists\rho\geq
0$ such that $|v_n| \leq \rho$, $\forall n\in\Natural$, then for any choice
of $\epsilon_n\geq \rho$ it is easy to verify that $\bm{a}_*\in
S_n[\epsilon_n]$, $\forall n\in\Natural$. A rigorous stochastic analysis in
the case of bounded noise, where almost sure convergence of the sequence of
estimates is proved for a special member of the rich family of the
\textit{Adaptive Projected Subgradient Method (APSM)}
\cite{YamadaOguraAPSMNFAO, KostasAPSMatNFAO, Sy.qne.siopt}, can be found in
\cite{simos.stochastic.analysis.spl}. In the case of unbounded noise, the
well-known Tchebichev inequality \cite{Loeve.1} suggests that for any
$\epsilon_n>0$,
\begin{equation*}
\prob\bigl\{\bm{a}_*\in S_n[\epsilon_n]\bigr\} = \prob \bigl\{
\bigl|\bm{u}_n^\top \bm{a}_* - y_n \bigr| \leq \epsilon_n \bigr\} \geq 1 -
\frac{\expect\bigl\{\bigl|\bm{u}_n^\top \bm{a}_* - y_n
  \bigr|^{2}\bigr\}}{\epsilon_n^2} = 1 -
\frac{\expect\bigl\{|v_n|^{2}\bigr\}}{\epsilon_n^2},
\end{equation*}
where $\prob$ denotes probability, and $\expect$ stands for the
expectation operator. In other words, $\epsilon_n$ defines also a
measure of \textit{confidence} in having the unknown $\bm{a}_*$ in the
hyperslabs \eqref{eq:Hyperslab}.

The (metric) projection mapping $P_{S_n[\epsilon_n]}$
\cite{bauschke.combettes.book} onto the hyperslab $S_n[\epsilon_n]$
\eqref{eq:Hyperslab} is given by the following simple analytic formula:
\begin{equation}
P_{S_n[\epsilon_n]}(\bm{a}) = \bm{a} +
\begin{cases} \frac{y_n - \epsilon_n - \bm{u}_n^{\top} \bm{a}}
  {\norm{\bm{u}_n}^2}\bm{u}_n, & \text{if}\ y_n - \epsilon_n > \bm{u}_n^{\top} \bm{a}
,\\ 0, & \text{if}\ |\bm{u}_n^{\top} \bm{a} - y_n |\leq
  \epsilon_n,\\ \frac{y_n + \epsilon_n - \bm{u}_n^{\top} \bm{a}}
          {\norm{\bm{u}_n}^2}\bm{u}_n, & \text{if}\ y_n + \epsilon_n <
          \bm{u}_n^{\top} \bm{a}.
\end{cases}
\label{project.hyperslab}
\end{equation}

In \cite{Kopsinis.Slavakis.TheodoridisIEEESP2011} sparsity was induced
within the convex analytic framework, and particularly via projections onto
convex $\ell_1$-balls. Here, the fixed point theoretical framework
\cite{GoebelKirk, bauschke.combettes.book} is used in order to generalize
the set theoretic estimation approach to support sparsity promoting
constraints, which do not lie under the umbrella of convexity. This is
realized via a novel operator theoretic framework, which embraces a wide
range of thresholding rules referred to as Generalized Thresholding (GT)
operators, described next.

\section{The Generalized Thresholding (GT) Mapping}\label{sec:GT}

A couple of definitions are necessary prior to introducing GT mapping.

\begin{definition}[The ordered tuple notation]\label{def:ord.tuple}
Given $K\in\overline{1,L}$, define the set of all \textit{ascending tuples
  of length $K$} as $\mathscr{T}(K,L)\coloneqq \{(l_1,l_2,\ldots,
l_K):1\leq l_1<l_2< \ldots < l_K \leq L\}$. Clearly, the cardinality of
$\mathscr{T}(K,L)$ is $\binom{L}{K}$. An example of such an ordered tuple
is the support of a vector $\bm{x}\in \Real^L$, defined by
$\supp(\bm{x})\coloneqq (l\in \overline{1,L}: x_l\neq 0)\in
\mathscr{T}(|\supp(\bm{x})|,L)$, where $|\cdot|$ stands for the cardinality
of a set.
\end{definition}

\begin{definition}[Subspace associated to a
    tuple]\label{def:special.subspace} Given $J\in\mathscr{T}(K,L)$,
  let $M_{J} \coloneqq \left\{\bm{a}\in\Real^L: a_l=0, \forall l\notin
  J\right\}$. Clearly, $M_J$ is a linear subspace of
  $\Real^L$. Moreover, notice that if $J_1 \subset J_2$, then $M_{J_1}
  \subset M_{J_2}$. In particular, if $\supp(\bm{x}_*) \subset J$,
  then $\bm{x}_*\in M_J$. An illustration of $M_J$ can be found in
  Fig.~\ref{fig:GT}.
\end{definition}

Motivated by the hard thresholding operator, let us introduce here the main
object of this study.

\begin{definition}[The mapping $T_{\text{GT}}^{(K)}$]\label{def:GT} Fix
  a positive integer $K\in \overline{1,L-1}$ and define
  $T_{\text{GT}}^{(K)}:\Real^L \rightarrow\Real^L$ as follows. For any
  $\bm{x}\in \Real^L$, the output $\bm{z}\coloneqq
  T_{\text{GT}}^{(K)}(\bm{x})$ is obtained according to the following
  steps:

\begin{enumerate}[leftmargin=0pt,itemindent=20pt]
\item\label{compute.tuple} Compute, first, the tuple $J_{\bm{x}}^{(K)}\in
  \mathscr{T}(K,L)$ which contains the indices of the $K$ largest, in
  absolute value, components of $\bm{x}$. To avoid any ambiguity, in the
  case where we identify more than one component of $\bm{x}$ with the same
  absolute value, we always choose the one with smallest index.
\item Define $\xi_{\bm{x}}^{(K)}\coloneqq \min \bigl\{|x_l|: l\in
  J_{\bm{x}}^{(K)}\bigr\}$. In words, $\xi_{\bm{x}}^{(K)}$ is the smallest
  among the $K$ largest absolute values of the components of
  $\bm{x}$. Clearly, $\forall l \notin J_{\bm{x}}^{(K)}$, $|x_l|\leq
  \xi_{\bm{x}}^{(K)}$.
\item\label{thres.components} Compute the components of $\bm{z}$ as:
  $z_l\coloneqq x_l$, if $l\in J_{\bm{x}}^{(K)}$, and $z_l\coloneqq
  \shr(x_l)$, if $l\notin J_{\bm{x}}^{(K)}$, where the function
  $\shr:\mathcal{D}_{\bm{x}}\rightarrow\Real$, with $\mathcal{D}_{\bm{x}}
  \coloneqq [-\xi_{\bm{x}}^{(K)},\xi_{\bm{x}}^{(K)}]$, satisfies the
  following properties:

\item\label{same.sign} $\tau\shr(\tau)\geq 0$, $\forall \tau\in
  \mathcal{D}_{\bm{x}}$. 
\item\label{shrinks} $|\shr(\tau)|\leq |\tau|$, $\forall
  \tau\in\mathcal{D}_{\bm{x}}$.
\item\label{strictly.shrinks} Going a step further than the previous
  property, we assume also that given any sufficiently small $\epsilon>0$,
  there exists a $\delta>0$, such that for any $\bm{x}\in \Real^L$, and
  $\forall \tau\in \mathcal{D}_{\bm{x}} \setminus (-\epsilon, \epsilon)$,
  $|\shr(\tau)| \leq |\tau|-\delta$. In other words, $\delta$ could be a
  user-defined parameter which guarantees that the function $\shr$ acts as
  a \textit{strict} shrinkage operator for all the components of $\bm{x}$
  with indexes not in $J_{\bm{x}}^{(K)}$. The $\epsilon$ parameter is
  introduced in order to exclude $0$ from the picture, since at this point
  the $\shr$ function usually takes the value of $0$, i.e., $\shr(0)=0$
  (see Fig.~\ref{fig:variousthresholds1}).
\end{enumerate}
\end{definition}

Put in other words, the GT mapping operates as follows; given the input
vector $\bm{x}$, a number of $K$ components of $\bm{x}$, i.e., those with
the $K$ largest absolute values, are kept intact, while the rest of them
are shrunk according to the $\shr$ function. See, for example,
Fig.~\ref{fig:GT}.

\begin{figure}[!th]
\centering
\includegraphics[width=0.3\textwidth]{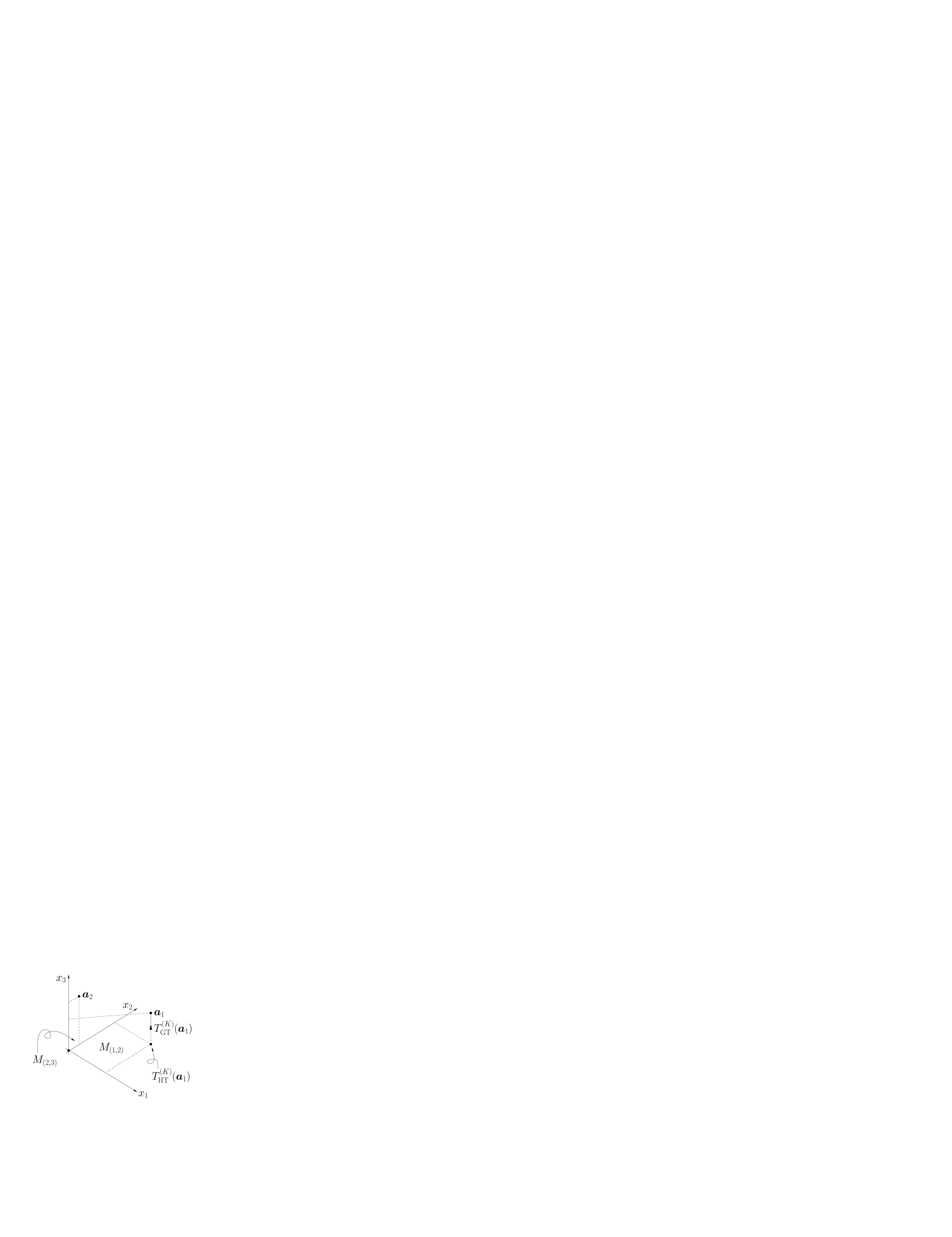}
\caption{An illustration of $T_{\text{GT}}^{(K)}$, for the case of a
  $3$-dimensional space, i.e., $L \coloneqq 3$, and $K \coloneqq 2$. Take
  for example the point $\bm{a}_1$. The $K=2$ largest, in magnitude,
  coordinates of $\bm{a}_1$ are the first two ones, i.e.,
  $J_{\bm{a}_1}^{(K)}= (1,2)$. The linear subspace $M_{(1,2)}$ stands for
  all those vectors in $\Real^3$ where all the components, except from
  those in the positions $(1,2)$, are equal to $0$. The first two
  components of $\bm{a}_1$ stay unaffected by $T_{\text{GT}}^{(K)}$, while
  the third one is shrinked by the function $\shr$. If this third
  coordinate is set to $0$, then $T_{\text{GT}}^{(K)}$ acts as the
  hard-thresholding mapping $T_{\text{HT}}^{(K)}$. On the other hand, the
  point $\bm{a}_2$ is already located in $M_{(2,3)}$, i.e., its first
  coordinate is $0$. Hence, the application of $T_{\text{GT}}^{(K)}$ to
  $\bm{a}_2$ has no effect, and $\bm{a}_2$ stays \textit{fixed} to its
  original position.}\label{fig:GT}
\end{figure}

\begin{figure*}[!th]
\centering
\includegraphics[width=1 \textwidth]{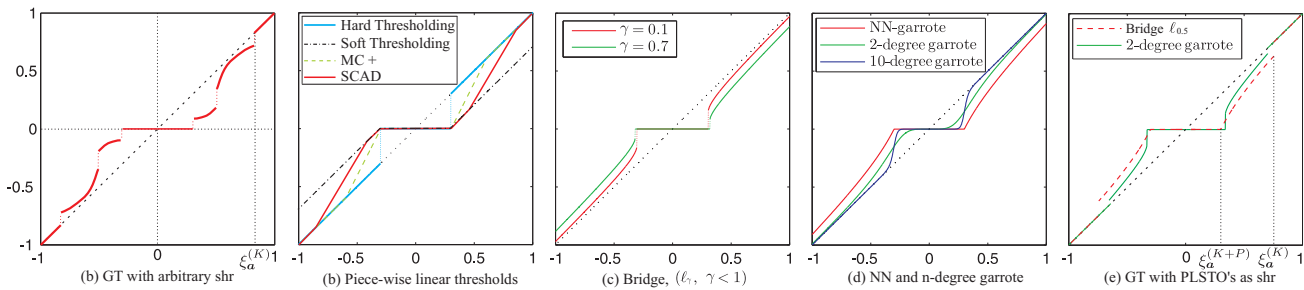}
\caption{Plots of Penalized Least-Squares Thresholding Operators (PLSTO)
  for various choices of the penalty function $p$ in
  \eqref{def.plsto}.}\label{fig:variousthresholds1}
\end{figure*}

The function $\shr$ is user-defined and it can get various forms as long as
it complies with the properties described before. As an example, a
thresholding operator in the GT family based on an arbitrary $\shr$
function is shown in Fig. \ref{fig:variousthresholds1}a. Note that it
comprises both discontinuities and nonlinear regions. A more systematic way
to built GT's is via the univariate \textit{Penalized Least Squares}
optimization task; given $\tilde{a}\in \Real$,
\begin{equation}
\min_{a} \frac{1}{2} \left(\tilde{a} - a\right)^2 +
  \lambda p(|a|), \label{def.upls}
\end{equation}
where $p(\cdot)$ is nonnegative, nondecreasing and differentiable function
on $(0,\infty)$. This problem is at the heart of many batch sparsity
promoting algorithms as it is discussed in App. \ref{sec:PLSTO}. It turns
out, that \eqref{def.upls} has, in general, a unique solution which is
obtained when $\tilde{a}$ is properly thresholded/shrinked
\cite{Antoniadis2007}. Accordingly, let us define the \textit{Penalized
  Least-Squares Thresholding Operator (PLSTO)} as the mapping which maps a
given $\tilde{a}$ to the previous unique minimizer:
\begin{equation}
T_{\text{PLSTO}}^{(p,\lambda)}: \tilde{a} \mapsto \argmin_{a\in \Real}
\frac{1}{2\lambda} \left(\tilde{a} - a\right)^2 +
  p(|a|). \label{def.plsto}
\end{equation}
In simple words, the PLSTO of \eqref{def.plsto} \textit{shrinks,} in some
sense that is dictated by $p$, the size of $\tilde{a}$. Examples of PLSTO's
exhibiting different characteristics are shown in
Fig.~\ref{fig:variousthresholds1}(b-d) and details together with the
corresponding literature review can be found in App. \ref{sec:PLSTO}.  All
the thresholding rules of Fig.~\ref{fig:variousthresholds1}(b-d) satisfy
the properties of Def.~\ref{def:GT}.\ref{same.sign} and
Def.~\ref{def:GT}.\ref{shrinks}. Moreover, they also satisfy the property
of Def.~\ref{def:GT}.\ref{strictly.shrinks} in their respective
strict-shrinkage region, i.e., in the case where the $\xi_{\bm{x}}^{(K)}$
lies in the domain of all those $\tau\in \Real$ such that
$\bigl|T_{\text{PLSTO}}^{(p)}(\tau)\bigr| < |\tau|$.  Notice, also, that we
do not impose any regularity conditions on $\shr$, like continuity or
differentiability, unlike most of the known PLSTO do \cite{Frank1993,
  Fan_Li2001, Zhang2010, Mazumder2011}. As a result, \textit{any} PLSTO,
i.e., \eqref{def.plsto}, can be used in the place of the $\shr$ function in
the GT operator. Examples of GT having PLSTO's as their $\shr$ function are
shown in Figs.~\ref{fig:variousthresholds1}a and
\ref{fig:variousthresholds1}e. Moreover, GT where the $\shr$ function is
the Bridge $\ell_{0.5}$ and the \textit{Smoothly Clipped Absolute Deviation
  Penalty (SCAD)} threshold are used and further discussed in the numerical
experiments section.

\section{The APGT Algorithm, Its Properties, and a Novel Operator Theoretic
Framework}\label{sec:algo}

\begin{algo}[The Adaptive Projection-based Generalized Thresholding (APGT)
    algorithm]\label{APGT} Given the user-defined sparsity level
  $K\in\overline{1,L-1}$, the sequence of non-negative parameters
  $(\epsilon_n)_{n\in\Natural}$, the number $q\in \Naturalstar$ of the
  hyperslabs to be processed concurrently at every time instant, the
  function $\shr$ for the generalized thresholding operation, and an
  arbitrary initial point, $\bm{a}_0\in \Real^L$, execute the following,
  for every $n\in \Natural$.

\begin{enumerate}[leftmargin=0pt,itemindent=20pt]
\item Define the sliding window $\mathcal{J}_n\coloneqq \overline{ \max\{0,
  n-q+1\}, n}$ on the time axis, of size at most $q$. The set
  $\mathcal{J}_n$ defines all the indices corresponding to the hyperslabs,
  which are to be processed at the time instant $n$. Among these, identify
  $\mathcal{I}_n\coloneqq \bigl\{i\in \mathcal{J}_n:
  P_{S_i[\epsilon_i]}(\bm{a}_n)\neq \bm{a}_n \bigr\}$, which correspond to
  the \textit{active} hyperslabs. Moreover, for every $i\in\mathcal{I}_n$,
  define the weight $\omega_i^{(n)}\coloneqq 1/|\mathcal{I}_n|$, where
  $|\mathcal{I}_n|$ denotes the cardinality of $\mathcal{I}_n$, in order to
  weigh uniformly the importance of the information carried by each
  hyperslab, $S_i[\epsilon_i]$. Other, more general, scenarios regarding
  the choice of $\{\omega_i^{(n)}\}_{i=1}^{|\mathcal{I}_n|}$ are also
  possible.

\item Collect the projections $P_{S_i[\epsilon_i]}(\bm{a}_n)$, $\forall
  i\in\mathcal{I}_n$ (see \eqref{project.hyperslab}).

\item Choose an $\varepsilon'\in (0,1]$, and let the
\textit{extrapolation parameter} $\mu_n$ take values from the interval
$[\varepsilon'\mathcal{M}_n, (2-\varepsilon')\mathcal{M}_n]$, where
\begin{subequations}\label{algo}
\begin{equation}
\mathcal{M}_n\coloneqq %
\begin{cases}
\frac{ \sum_{i\in\mathcal{I}_n} \omega_i^{(n)}
 \norm{P_{S_i[\epsilon_i]}(\bm{a}_n) - \bm{a}_n}^2} {
 \norm{\sum_{i\in\mathcal{I}_n} \omega_i^{(n)}
   P_{S_i[\epsilon_i]}(\bm{a}_n) - \bm{a}_n}^2}, & \\
& \hspace{-100pt}\text{if}\ \sum_{i\in\mathcal{I}_n} \omega_i^{(n)}
   P_{S_i[\epsilon_i]}(\bm{a}_n) \neq \bm{a}_n,\\
1, & \hspace{-100pt} \text{otherwise}. \label{Mn}
\end{cases}
\end{equation}
Notice that due to the convexity of the function $\norm{\cdot}^2$, we
always have $\mathcal{M}_n\geq 1$. As such, the parameter $\mu_n$ takes
values larger than or equal to $2$. In general, the larger the $\mu_n$, the
larger the convergence speed of the proposed algorithm.

\item Compute the next estimate by
\begin{equation}
\bm{a}_{n+1}\coloneqq
\begin{cases}
T_{\text{GT}}^{(K)} \Biggl(\bm{a}_n + \mu_n
\biggl(\sum\limits_{i\in\mathcal{I}_n} \omega_i^{(n)}
P_{S_i[\epsilon_i]}(\bm{a}_n) - \bm{a}_n\biggr)\Biggr), &
\\ & \hspace{-50pt}\text{if}\ \mathcal{I}_n\neq\emptyset,\\
T_{\text{GT}}^{(K)}(\bm{a}_n), & \hspace{-50pt} \text{if}\ \mathcal{I}_n =
\emptyset.
\end{cases}\label{main.recursion}
\end{equation}
\end{subequations}
\end{enumerate}
\end{algo}

In order to theoretically support the incorporation of GT into parameter
estimation schemes, a novel family of mappings, called the
\textit{partially quasi-nonexpansive mappings}, which, to the best of our
knowledge, appears for the first time in the related literature
\cite{bauschke.combettes.book}. The reasons for defining this new class of
mappings are: (i) this family includes as a special case the previously
defined generalized thresholding operator $T_{\text{GT}}^{(K)}$, and, thus,
it establishes a general theoretical framework for sparsity-promoting
mappings, (ii) it introduces sound theoretical tools, which help to attack
\textit{non-convexly} constrained learning problems, and (iii) it
generalizes the very recent results, obtained for the \textit{Adaptive
  Projected Subgradient Method (APSM)} \cite{Sy.qne.siopt}, to
\textit{non-convexly} constrained online learning tasks (see
App.~\ref{sec:analysis.algo}).

Although the following discussion can be naturally extended to general
Hilbert spaces, for the sake of simplicity we focus here on the Euclidean
space $\Real^L$, i.e., $T: \Real^L \rightarrow \Real^L$. A concept of
fundamental importance, associated to every mapping $T$, is its
\textit{fixed point set} $\Fix(T) \coloneqq \bigl\{\bm{a}\in \Real^L:
T(\bm{a})=\bm{a}\bigr\}$ \cite{GoebelKirk,
  bauschke.combettes.book}. In other words, $\Fix(T)$ reveals the hidden
modes of $T$, by putting together all those points unaffected by $T$. To
leave no place for ambiguity, every $\Fix(T)$ that appears in the sequel is
assumed nonempty.

\begin{definition}[The class of partially quasi-nonexpansive
    mappings]\label{def:pqne} A mapping $T$ is called \textit{partially
    quasi-nonexpansive,} if
\begin{equation}\label{pqne}
\begin{gathered}
\forall \bm{x} \in \Real^L, \exists Y_{\bm{x}}
  \subset \Fix(T): \forall \bm{y}\in Y_{\bm{x}},\\ \norm{T(\bm{x})-\bm{y}}
  \leq \norm{\bm{x}-\bm{y}}.
\end{gathered}
\end{equation}
The fixed point set $\Fix(T)$ is \textit{not}
  necessarily a convex set. Let us also define a stronger version of
  \eqref{pqne}; the mapping $T$ will be called \textit{strongly} or
  \textit{$\eta$-attracting partially quasi-nonexpansive mapping} if there
  exists an $\eta>0$ such that
\begin{equation} \label{spqne}
\begin{gathered}
\forall \bm{x} \in \Real^L, \exists Y_{\bm{x}} \subset \Fix(T):
\forall \bm{y}\in Y_{\bm{x}}, \\ \eta \norm{\bm{x} -T(\bm{x})}^2 \leq
\norm{\bm{x} - \bm{y}}^2 - \norm{T(\bm{x}) - \bm{y}}^2.
\end{gathered}
\end{equation}
\end{definition}

An example of such a mapping \eqref{spqne} is the novel generalized
thresholding mapping of Section \ref{sec:GT} (for a proof see
App.~\ref{sec:properties.T_K}). In App.~\ref{sec:properties.T_K}, we will
also verify that $\Fix(T_{\text{GT}}^{(K)})$ is a union of subspaces, which
is indeed a non-convex set. Recall that at the heart of any
sparsity-promoting learning method lies the search for a solution in a
union of subspaces \cite{Lu.Do.2008, eldar.michali.union.2009}. It must be
pointed out that a number of well-known mappings, e.g.,
\cite{bauschke.combettes.book, Combettes.Pesquet.proximal.2010,
  myyy.soft.thres.icassp10}, are special cases of the previously defined
class of partially quasi-nonexpansive ones.

The convergence analysis of the APGT is given by the following
Thm.~\ref{thm:algo}. This analysis is based on a set of deterministic
assumptions, given below. Since the APGT is based on the mapping
$T_{\text{GT}}^{(K)}$, whose fixed point set (see
App.~\ref{sec:properties.T_K}) is non-convex, this is the first time that
the results of \cite{YamadaOguraAPSMNFAO, KostasAPSMatNFAO, Sy.qne.siopt}
are generalized to \textit{non-convexly constrained online learning tasks.}

\begin{assumption}\label{assumptions}\mbox{}
\begin{enumerate}[leftmargin=0pt,itemindent=20pt]

\item\label{ass:Omega.n} Assume that $\exists n\in \Natural$ such that
  $\Omega_n \coloneqq M_{J_{\bm{a}_n}^{(K)}} \cap
  \bigcap_{i\in\mathcal{I}_n} S_i[\epsilon_i] \neq \emptyset$. Let us
  explain here the physical reasoning behind this assumption. Recall, here,
  that $\{S_i[\epsilon_i]\}_{i\in\mathcal{I}_n}$ is the set of all active
  hyperslabs (see Alg.~\ref{APGT}), at the time instant $n$. For an
  appropriate choice of the parameters $(\epsilon_n)_{n\in\Natural}$ (see
  \eqref{eq:Hyperslab}), the hyperslabs contain the desired $\bm{a}_*$ with
  high probability. Moreover, as time goes by, and due to a long sequence
  of projections in \eqref{algo}, the orbit $(\bm{a}_n)_{n\in\Natural}$ is
  attracted closer and closer to the hyperslabs; and as a consequence,
  closer to $\bm{a}_*$. For this reason, it is natural to expect that
  $\supp(\bm{a}_n)$ is similar to $\supp(\bm{a}_*)$, and hence
  $M_{J_{\bm{a}_n}^{(K)}}$ to $M_{J_{\bm{a}_*}^{(K)}}$, at some time
  $n$. Since $M_{J_{\bm{a}_*}^{(K)}}$ enjoys a non-empty intersection with
  $\bigcap_{i\in\mathcal{I}_n} S_i[\epsilon_i]$, with high probability, we
  anticipate that the same also happens to $M_{J_{\bm{a}_n}^{(K)}}$.

\item\label{ass:finite.Omega} Assume that there exists a time instant
  $n_0\in\Natural$, and an $N\in \Naturalstar$, such that
  $\bigcap_{n=n_0}^{n_0+N-1} \Omega_n \neq \emptyset$.

\item\label{ass:Omega} Assume that $\Omega \coloneqq
  \liminf_{n\rightarrow\infty}\Omega_n \coloneqq \bigcup_{n\geq 0}
  \bigcap_{m\geq n} \Omega_m \neq \emptyset$. In other words, we assume
  that the set of all points, which belong to all but a finite number of
  $\Omega_n$s, is nonempty.

\end{enumerate}
\end{assumption}

\begin{theorem}[Properties of the APGT]\label{thm:algo}\mbox{}

\begin{enumerate}[leftmargin=0pt,itemindent=20pt]

\item\label{simple.monotinicity} Let Assumption
  \ref{assumptions}.\ref{ass:Omega.n} hold true. Then, $d(\bm{a}_{n+1},
  \Omega_n) \leq d(\bm{a}_n, \Omega_n)$, where $d(\cdot,\Omega_n)$ stands
  for the (metric) distance function \cite{bauschke.combettes.book} to
  $\Omega_n$.

\item\label{bound} Let Assumption
  \ref{assumptions}.\ref{ass:finite.Omega} hold true. Then,
\begin{align*}
d^2 \Bigl(\bm{a}_{n_0+N}, \bigcap_{n=n_0}^{n_0+N-1} \Omega_n \Bigr) & \leq
d^2 \Bigl(\bm{a}_{n_0}, \bigcap_{n=n_0}^{n_0+N-1} \Omega_n \Bigr) \\
& \hspace{-70pt}- \frac{(\varepsilon')^2}{q} \sum_{n=n_0}^{n_0+N-1}
\max\bigl\{d^2(\bm{a}_n,S_j[\epsilon_j]): j\in\mathcal{J}_n\bigr\}.
\end{align*}
In other words, the previous inequality establishes a bound on the distance
of the estimates from a \textit{finite} intersection of the $\Omega_n$s. If
we assume, also, that there exists an estimate $\bm{a}_n$ which does not
belong to such an intersection, i.e., $\exists n'\in
\overline{n_0,n_0+N-1}$ such that
$\max\bigl\{d^2(\bm{a}_{n'},S_j[\epsilon_j]): j\in\mathcal{J}_{n'}\bigr\}
>0$, then the previous result claims that the APGT forces $\bm{a}_{n_0+N}$
to be located \textit{strictly} closer to $\bigcap_{n=n_0}^{n_0+N-1}
\Omega_n$ than $\bm{a}_{n_0}$ is.

\item Let Assumption
  \ref{assumptions}.\ref{ass:Omega} holds true. Then,

  \begin{enumerate}[leftmargin=0pt,itemindent=20pt]

  \item\label{thm:exist.cluster} the set of all cluster points of the
    sequence $(\bm{a}_n)_{n\in\Natural}$ is nonempty, i.e.,
    $\mathfrak{C}\bigl((\bm{a}_n)_{n\in\Natural} \bigr) \neq
    \emptyset$.

  \item\label{thm:distance.goes.2.zero}
    $\lim_{n\rightarrow\infty}d(\bm{a}_n,S_n[\epsilon_n]) = 0$. In other
    words, as the time advances, the orbit $(\bm{a}_n)_{n\in\Natural}$
    approaches $(S_n[\epsilon_n])_{n\in\Natural}$.

    \item\label{thm:inclusion.4.cluster.points}
      $\mathfrak{C}\bigl((\bm{a}_n)_{n\in\Natural}\bigr) \subset\Fix
      \bigl(T_{\text{GT}}^{(K)}\bigr) =
      \bigcup_{J\in\mathscr{T}(K,L)}M_J$. In words, the APGT generates a
      sequence of estimates $(\bm{a}_n)_{n\in\Natural}$, whose
      \textit{cluster points are sparse vectors,} of sparsity level no
      larger than $K$.

  \end{enumerate}
\end{enumerate}
\end{theorem}

\begin{proof}
See Appendix \ref{sec:analysis.algo}.
\end{proof}

\section{Numerical Experiments}\label{sec:simulations}

In this section, our main intention is to provide the proof of concept of
the theoretical findings presented in Thm.~\ref{thm:algo}. This is realized
via the performance evaluation of \eqref{main.recursion}, where the
shrinkage function $\shr$, in Def.~\ref{def:GT}, assumes any form of
$T_{\text{PLSTO}}^{(p)}$, defined in \eqref{def.plsto}. This study is not
meant to be exhaustive, and in order to demonstrate the potential of the
proposed technique, the hard thresholding (HT) as well as the PLSTOs
corresponding to the SCAD \cite{Fan_Li2001} and the $\ell_{\gamma}$ penalty
($\gamma<1$) \cite{Frank1993} are examined, since they exhibit distinct
characteristics, as it is illustrated in
Figs.~\ref{fig:variousthresholds1}b and \ref{fig:variousthresholds1}e,
respectively. Notice that the associated penalty functions are
non-convex. The resulting thresholding rules are called the \textit{SCAD}
and the \textit{Bridge Thresholding (BT)}, respectively. Notice, also, that
SCAD is a piece-wise linear thresholding operator, whereas, the BT exhibits
strong discontinuity and non-linearity. 

In order to comply with the theory, the SCAD, the BT, and the HT are used
as shrinkage functions $\shr$ in Def.~\ref{def:GT}, for all the components
$x_i$ with index $i \notin J_{\bm{x}}^{(K)}$, where $K$ stands for an
estimate of the true $K_* \coloneqq \norm{\bm{a}_*}_0$. To this end, we
have slightly modified the classical SCAD, BT, and HT rules in order to fit
our need to keep a number of $K$ components of a vector intact. As such,
the SCAD thresholding operates according to the following rule; given the
input $\bm{x}\in \Real^L$ and the output vector $\bm{z}:=
T_{\text{SCAD}}(\bm{x})$, the $i$-th coordinate of $\bm{z}$, where $i\notin
J_{\bm{x}}^{(K)}$, is given by the next rule:
\begin{equation}
\label{eq:SCAD}
z_i=
\begin{cases}
0, & \text{if}\ |x_i| \le \lambda, \\
\sign(x_i)(|x_i|-\lambda - \delta)_+, & \text{if}\ |x_i| \in
\bigl(\lambda, 2\lambda \bigr], \\
\sign(x_i) \bigl(\frac{(\alpha-1)|x_i|-\alpha \lambda}{\alpha -2} -
  \delta\bigr)_+, & \\
& \hspace{-60pt}\text{if}\ |x_i| \in \bigl(2\lambda,
\min\bigl\{\xi_{\bm{x}}^{(K)},\alpha \lambda\bigr\}\bigr],
\end{cases}
\end{equation}
where $\lambda$ is the regularization parameter, which appears in the
definition of the PLSTO in \eqref{def.plsto}, $\alpha$ is a user-defined
parameter, inherent to SCAD \cite{Fan_Li2001}, $\delta>0$ is a sufficiently
small user-defined parameter motivated by
Definition~\ref{def:GT}.\ref{strictly.shrinks}, and $(\cdot)_+ \coloneqq
\max\{0,\cdot\}$, introduced here in order to leave no place for
ambiguities. Our modification on the classical SCAD can be seen by the
introduction of $\delta$, $\xi_{\bm{x}}^{(K)}$, and $(\cdot)_+$.

Similarly, given the classical version of the BT rule
\cite{Antoniadis2001}, our modified BT is given as follows by involving the
quantity $\xi_{\bm{x}}^{(K)}$ in the computations: $\forall
i\notin J_{\bm{x}}^{(K)}$, 
\begin{equation}
z_i=
\begin{cases}
\sign(x_i)(\bar{z}_i -\delta)_+, & \\
& \hspace{-70pt} \text{if}\ \min
\bigl\{c_{\text{BT}}(\lambda, \gamma), \xi_{\bm{x}}^{(K)} \bigr\} \leq
|x_i| \leq  \xi_{\bm{x}}^{(K)},\\ 
0, & \hspace{-70pt} \text{otherwise},
\end{cases}
\end{equation}
where $\lambda$ is the corresponding regularization parameter in
\eqref{def.plsto}, $\gamma\in (0,1)$ is a user-defined parameter, and
\begin{equation*}
c_{\text{BT}}(\lambda, \gamma) \coloneqq \left(-\frac{1}{\lambda
  \gamma(\gamma-1)}\right)^{\frac{1}{\gamma-2}} \hspace{-10pt} + \lambda
\gamma \left(-\frac{1}{\lambda
  \gamma(\gamma-1)}\right)^{\frac{\gamma-1}{\gamma-2}}.
\end{equation*}
The term $\bar{z}_i$ stands for the solution of the equation $\bar{z}_i +
\sign(z_i) \lambda \gamma \bar{z}_i^{\gamma-1}=|x_i|$. When $\gamma$ is set
equal to $0.5$, $\bar{z}_i$ is obtained in closed form by solving a third
order polynomial equation. Similarly, HT is given by the
following rule; $\forall i\notin J_{\bm{x}}^{(K)}$,
\begin{equation*}
z_i= \begin{cases}
0, & \text{if}\ |x_i| \leq \min\bigl\{ \lambda, \xi_{\bm{x}}^{(K)}
\bigr\},\\  
\sign(x_i)(|x_i| - \delta)_+, & \text{otherwise},
\end{cases}
\end{equation*}
where the $\lambda$ is introduced here in order to be compliant also to a
definition of the HT used often in the literature (see the discussion in
Appendix~\ref{sec:PLSTO}).

In the following experiments, unless otherwise stated, the signal under
consideration has $L=1024$ and $K_* = 100$. Moreover, the classical CS
signal recovery problem is considered, where the input (sensing) vectors
have independent components drawn from a normal distribution
$\mathcal{N}(0,1)$, and the observations are corrupted by additive white
Gaussian noise of variance $\sigma^2=0.1$. Regarding APGT, the
extrapolation parameter $\mu_n$ is set equal to $\mathcal{M}_n$, and the
hyperslab parameter $\epsilon_n\coloneqq 1.3\sigma$, $\forall n$. In this
paper, for all the techniques employed, configurations leading to the
fastest convergence rate are of principal interest. From this perspective,
unless otherwise stated, $q$ is fixed to $390$ since this appeared to be
the lowest $q$ value leading to enhanced convergence speed for the specific
$L$ and $K$ values. It should be stressed out that the method is not
sensitive to the parameter $q$. A larger $q$ value would only add to
computational complexity without any significant contribution to
performance. An extensive and complementary experimental study of the APGT
performance, in the case where $q$ is confined to small values, which
relates to very low computational complexity techniques, can be found in
\cite{Kopsinis.GT.low.complexity, KopsinisISCAS2013}. In all of the
succeeding figures, the MSE stands for $\text{MSE}_n \coloneqq
\frac{1}{\tau L} \sum_{i=1}^{\tau} \norm{\bm{a}_* - \bm{a}_n(i)}^2$, where
$(\bm{a}_n(i))_{n\in\Natural}$ is the sequence generated by the $i$-th
realization of Alg.~\ref{APGT}, and $\tau\coloneqq 100$ is the number of
independent realizations in order to smooth out the obtained performance
curves.

\subsection{Employing time-invariant thresholding operators}
\label{sec:performance_fixedPLSTO}

By the modifier ``time-invariant'', we mean that the user-defined parameter
$\lambda$ in \eqref{def.plsto} remains fixed for all the time instants
$n\in \Natural$. The performance of all the employed methods is given in
Fig.~\ref{fig:constthres}. In all cases, $K \coloneqq K_*$. The
regularization parameter $\lambda$ was optimized leading to the values
shown in the corresponding figure legend. Moreover, APGT-SCAD, without
being considerably sensitive to parameter $\alpha$, appeared to perform
best when adopting the relatively large value $\alpha=12$.

\begin{figure}[!th]
\centering
\subfloat[Time-invariant thresholding, i.e., fixed $\lambda$.]
         {\includegraphics[width=0.45\textwidth]{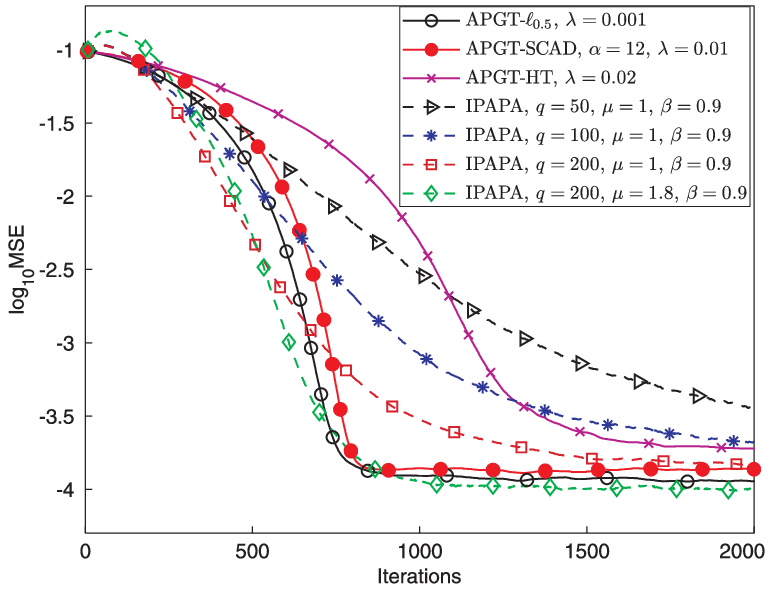}
\label{fig:constthres}}\hfil
\subfloat[Time-adaptive
  thresholding, i.e., time-varying
  $\lambda$.]{\includegraphics[width=0.45\textwidth]
  {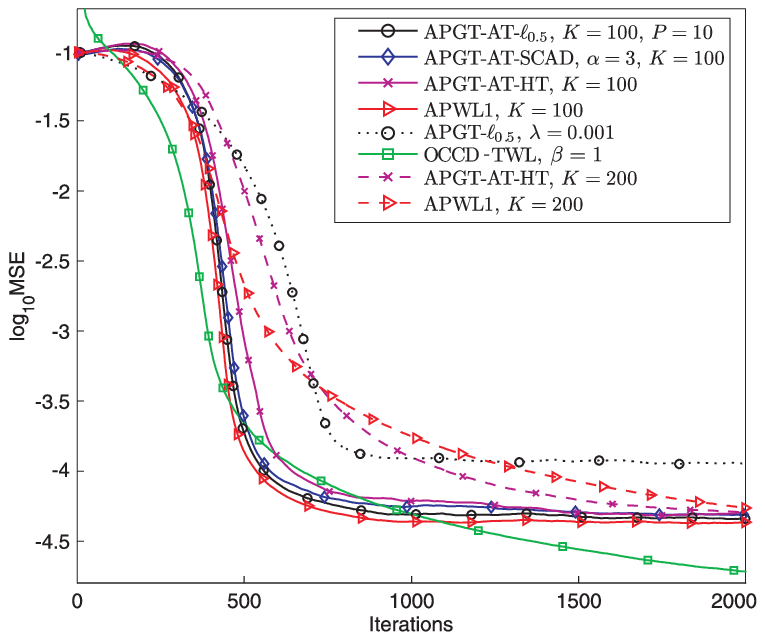}
                \label{fig:adaptthres390}}
        \caption{(a) Performance study of APGT using thresholding operators
          which are fixed in each iteration and comparison with IPAPA
          algorithm. (b) Performance study of APGT using thresholding
          operators which are changing in each iteration and comparison
          with LASSO solution.}
\end{figure}

For comparison, the \textit{Improved Proportionate Adaptive Projection
  Algorithm (IPAPA),} described in \cite{Hoshuyama_ipapa_2004,
  paleologu2010efficient}, is employed. The \textit{projection order} of
the IPAPA, which plays a similar role to $q$, and therefore the same
notation is used, is the major factor which dictates its
performance. Dashed curves indicated with triangles, stars and squares
correspond to values of $q$ equal to $50$, $100$, and $200$,
respectively. The step parameter of the IPAPA is denoted by $\mu$. The best
IPAPA performance, i.e., the one depicted with a dashed curve with
diamonds, is achieved with $q=200$ and $\mu=1.8$. For lower $q$ values,
such a large $\mu$ led to unstable performance. In all cases, the parameter
$\beta$, which tunes the weights in the proportionate
algorithm\footnote{See parameter $\alpha$ in $(2)$ of
  \cite{paleologu2010efficient}. We call it here $\beta$ in order to avoid
  confusion with the parameter $\alpha$ of SCAD.}, was given the large
value $\beta=0.9$ in order to exhibit enhanced sparsity promoting
behavior. When larger $q$ values are used, e.g., $q=400$, the performance
turned to become somewhat faster, but with a quite elevated steady-state
error floor, so the corresponding performance curves are not
shown. Moreover, a set-membership counterpart of IPAPA \cite{Werner.Diniz}
was also examined. This algorithm performed similarly to IPAPA, so the
results are not shown to ease visualization. It is clear that the
APGT-$\ell_{0.5}$ performs as well as IPAPA. However, this is achieved
under a significantly lower computational burden, as will be discussed in
Section~\ref{sec:complexity}.

\subsection{Employing time-adaptive thresholding operators}
\label{sec:performance_adaptPLSTO}

In the previous section, the exact shape of the thresholding function was
determined in advance using \textit{fixed} values for the associated
parameters, e.g., $\lambda$, $\gamma$, $\alpha$, etc. This is quite
limiting, since the proposed technique has the potential to incorporate
time-adaptive a-priori information, in the form of time-varying
thresholding operators. This section demonstrates that exploiting this
freedom leads APGT to enhanced performance. In particular, $\lambda$ in
\eqref{def.plsto} changes as time $n$ advances. In order to explicitly
describe this dependency of $\lambda$ to $n$, we will use hereafter the
notation $\lambda_n$. Assuming that an estimate $K$ of the true sparsity
level $K_*$ is available at each iteration $n$, parameter $\lambda_n$ is
properly tuned in order to guarantee that after thresholding, a fixed
number of components will be set equal to zero. With respect to the HT
operator, in order to achieve a sparsity level equal to $K$, i.e., $L-K$
components are zero, the quantity $\lambda_n$ should be set equal to
$\xi_{\bm{a}_n}^{(K)}$, $\forall n$. For the SCAD case, $\lambda_n
\coloneqq \frac{1}{\alpha} \xi_{\bm{a}_n}^{(K)}$, $\forall n$,
(refer to \eqref{eq:SCAD}). In this way, the SCAD shrinkage behavior is
preserved and tuned by the user-defined parameter $\alpha$. In a similar
manner, an adaptive BT can be built. Going even further, apart from the $K$
larger in magnitude components which remain unaltered, the next, say $P$,
smaller in magnitude components could be shrunk according to the bridge
rule. This is achieved if we notice that, by definition,
$\xi_{\bm{a}_n}^{(K+P)} \leq \xi_{\bm{a}_n}^{(K)}$, $\forall P\in
\overline{1,L-K}$, and that the parameter $\lambda_n$ is defined here as
the solution of the following equation $\xi_{\bm{a}_n}^{(K+P)} =
c_{\text{BT}}(\lambda_n, \gamma)$. In particular, for $\gamma=0.5$, this
solution obtains a closed form:
\begin{equation}
\label{eq:bridgelambda_n}
\lambda_n = 4\left(\frac{\xi_{\bm{a}_n}^{(K+P)}}{3}\right)^{\frac{3}{2}},
\quad \forall n.
\end{equation}
For convenience, the full GT operator involving the $\ell_{0.5}$ shrinkage
is given next: $\forall i\notin J_{\bm{x}}^{(K)}$,
\begin{equation}
z_i=
\begin{cases}
0, & \text{if}\ |x_i| \le \xi_{\bm{a}_n}^{(K+P)}, \\
\sign(x_i)(\bar{z}_i -\delta)_+, &
\text{if}\ \xi_{\bm{a}_n}^{(K+P)} < |x_i| \le \xi_{\bm{a}_n}^{(K)},
\end{cases}
\end{equation}
where $\bar{z}_i$ satisfies $\bar{z}_i+\frac{1}{2 \sqrt{\bar{z}_i}}
\lambda_n \sign(z_i)=|x_i|$, and $\lambda_n$ is given by
\eqref{eq:bridgelambda_n}.

The performance of APGT methods, using the previous time-adaptive
thresholding strategy, hereafter abbreviated as APGT-AT, is shown in
Fig.~\ref{fig:adaptthres390}. For reference, the dotted curve marked with
open circles is the one from Fig. \ref{fig:constthres} corresponding to the
best APGT method with a fixed $\lambda$. Moreover, the best results for the
APGT-AT-$\ell_{0.5}$ are obtained when $P$ assumes a small integer value,
such as $10$. A conclusion that can be easily drawn is that the
incorporation of adaptive thresholding led to a performance
boost. Moreover, the performance achieved depends on the thresholding
operator that is adopted, with the BT leading to somewhat faster
convergence speed compared to SCAD and HT. The performance of APWL1,
proposed in \cite{Kopsinis.Slavakis.TheodoridisIEEESP2011}, is also shown
with solid line marked with triangles. It appears that the newly proposed
algorithms, and especially APGT-AT-$\ell_{0.5}$, succeeds in achieving a
similar convergence behavior and speed compared to APWL1 and, as it will be
discussed in Section~\ref{sec:complexity}, with half the computational
complexity. For completeness, the \textit{Online Cyclic Coordinate Descent
  - Time Weighted Lasso (OCCD-TWL),} presented in
\cite{Angelosantejournal2010}, is depicted with solid line marked with
squares. The latter is an online algorithm approximating the LASSO problem
solution. It is observed, that APGT ($q=390$), demonstrates a performance
competitive to OCCD-TWL, which is an $\mathcal{O}(L^2)$ complexity
algorithm.

The advantages of the APGT algorithm over the APWL1 are not limited to the
performance improvements and/or to computational complexity savings. The
proposed theoretical framework is general enough in order to include other
thresholding operators as well, either existing or newly defined. However,
the scope of this paper is not a simulation study of all these alternatives
of thresholding, and such a route will be studied elsewhere. For example,
in \cite{Kopsinis.GT.low.complexity}, implementations of the proposed
scheme driven by a different set of PLSTOs, suitable for low complexity
operation, and a novel specially customized thresholding operator are
presented. In that case, comparison with linear complexity sparsity
inducing algorithms, such as the Reweighted Zero Attracting-Least Mean
Square (RZA-LMS) \cite{ChenHero09}, $\ell_0$-LMS
\cite{Su_L0_LMS_performance_IEEESP2012}, and the Sparse Adaptive Orthogonal
Matching Pursuit (SpAdOMP) \cite{mileounisspadomp} is made in more advanced
scenarios, such as system identification with correlated input signal (see
\cite{Kopsinis.GT.low.complexity}) and sparse signal estimation corrupted
by non-symmetric and/or impulsive noise.

\subsection{Robustness against inaccurate sparsity level estimates}
\label{sec:robustness}

With the aid of Fig.~\ref{fig:wrongsparsity20}, the effect of over- and
under-estimation of $K_*$ is discussed for the reduced complexity case of
$q=20$. We choose a low value for $q$, since we noticed that such a
scenario reveals more distinctly the performance sensitivity and related
behavior of the APGT with over- or under-estimations of $K_*$. Moreover,
the use of a low value of $q$, reveals the performance advantages of the
GT, compared to other linear complexity algorithms, such as the
$\ell_0$-LMS \cite{Su_L0_LMS_performance_IEEESP2012}. As it is seen from
the Fig.~\ref{fig:wrongsparsity20}, the use of the GT mapping results in
enhanced performance w.r.t.\ both APWL1 and $\ell_0$-LMS, where the latter
was fine-tuned for best convergence speed/error floor trade off. In order
to have a reference of the performance achieved when the true sparsity
level is given, the APGT-AT-$\ell_{0.5}$ with $K=100$, is also provided in
Fig.~\ref{fig:wrongsparsity20}. Let us start with the under-estimation case
and assume that $K=80$, i.e., $20\%$ lower compared to the true sparsity
level. Let us take, for example, the APGT-AT-SCAD curve, which shows an
elevated error floor. Notice that the case of under-estimations of $K_*$ is
not supported theoretically by Thm.~\ref{thm:algo}. With respect to
over-estimation, APGT is shown to be very robust. For example, let us see
the case where $K_*$ is over-estimated by $100\%$, i.e., $K\coloneqq
2K_*$. The performance achieved by APGT-AT-$\ell_{0.5}$ (solid line with
open circles) is still much better compared to the APWL1, even if APWL1
uses an accurate estimate for the $K_*$. Moreover, the degradation resulted
from such a large over-estimation appears to be limited. Remarkably, in
this low $q$ case, both APGT-AT-HT and APGT-AT-SCAD, drawn with solid lines
marked with x-crosses and diamonds, respectively, have benefited from the
over-estimation. The reason for this is that when $q$ is small, the
tentative estimates of the unknown vector in each iteration are likely to
be not accurate enough in order for the $K_*$ larger of them to reveal the
true support of the vector. An over-estimated $K_*$ leads to less strict HT
and SCAD thresholding operators, which allow components that would
otherwise be set equal to zero, to survive. All the results above have been
confirmed with higher levels of over-estimation.

\begin{figure}[!th]
\centering
\subfloat[Robustness against erroneous estimates of the sparsity
  level.]{\includegraphics[width=0.45\textwidth]{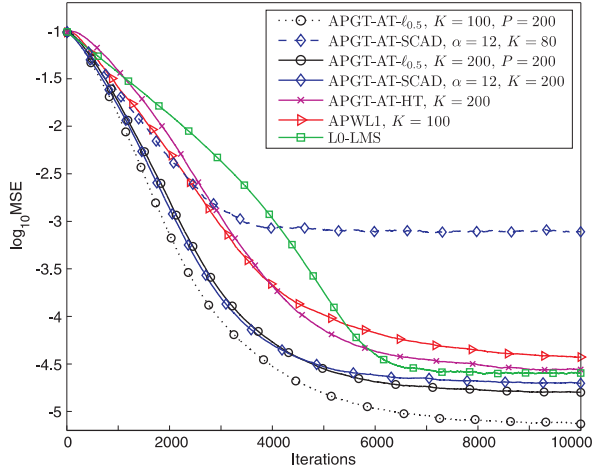}
 \label{fig:wrongsparsity20}}\hfil
\subfloat[Robustness against time variations of the desired solution.]
{\includegraphics[width=0.45\textwidth]{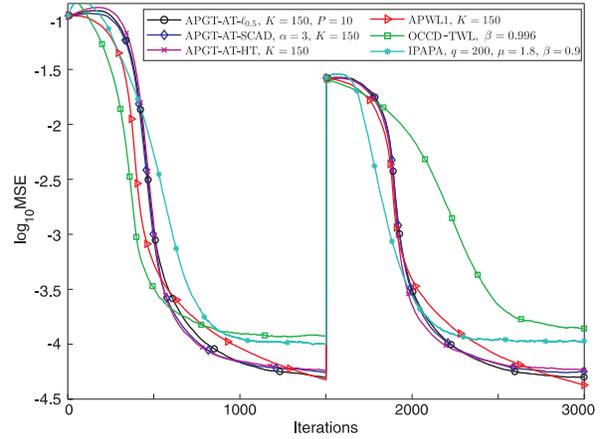}
\label{fig:timevar}}
\caption{(a) Robustness of APGT-AT in the cases of an under-estimation and
  an over-estimation of $K_*=100$, i.e., $K=80$ and $K=200$,
  respectively. The $q=20$ in these experiments. (b) The unknown
  vector has a sparse wavelet representation which changes abruptly after
  the reception of 1500 observations.}
\end{figure}

The results are similar when the algorithms operate with higher complexity,
i.e., $q=390$, with the difference that the performance of APGT-AT-HT and
APGT-AT-SCAD does not benefited as much as previously by an over-estimation
of $K$. The APGT-AT-SCAD and APGT-AT-$\ell_{0.5}$ perform similarly, so the
corresponding curves are not shown. A thorough examination of several
scenarios, in the case where $q$ attains low values, is deferred to a
future work.

\subsection{Tracking ability of the APGT} \label{sec:timevar}

Fig.~\ref{fig:timevar} shows the ability of the tested algorithms to track
an abrupt change of the unknown vector $\bm{a}_*$, which is realized here
after $1500$ observations is examined. This is a typical setting used in
 adaptive filtering \cite{Haykin} community to study the tracking
agility of an algorithm. Here, in order to give an essence from the CS
paradigm, we consider the vector to be not sparse itself but to have a
sparse wavelet representation. In the first half, the signal under
consideration is of length $L=1024$, with $K_*=100$ non zero wavelet
coefficients. However, at the $1500$ time instant, ten randomly selected
wavelet coefficients change their values from $0$ to a randomly selected
nonzero one. Since the sparsity level of the signal changes (from $100$ to
$110$, at most) and it is not possible to know $K_*$ exactly in advance,
taking into account that the methods we propose are quite robust to $K_*$
over-estimations, we set $K = 150$ throughout the whole
experiment. Moreover, $q$ is set to $390$.

For the OCCD-TWL, an RLS-like forgetting factor lower than $1$ is adopted, in
order to succeed in re-estimating the unknown signal after the abrupt
change. More specifically, the value of $0.996$ appeared to offer a good
trade-off between convergence speed and steady-state error floor. However,
the OCCD-TWL convergence speed slows down after the $1500$ time instant,
something which was observed and discussed in
\cite{Kopsinis.Slavakis.TheodoridisIEEESP2011} as well. The IPAPA method,
catches up quickly after the abrupt change; however, the attained error
floor is higher than that of the APGT.

\subsection{Computational complexity}\label{sec:complexity}

The choice of the thresholding operator affects significantly the overall
computational burden for two reasons. First, the thresholding function
itself requires a larger or smaller number of mathematical operations
depending on the specific thresholding rule. Such operations can be
multiplications, divisions, as well as sorting operations. Additions are
ignored since they are considered to be much less costly. A second
attribute of the thresholding rule, which affects complexity, is whether
its outcome is a sparse vector with a certain sparsity level or
not. Indeed, if the thresholding operator produces vectors which are, say,
$\overline{K}$-sparse, then projections in APGT involve inner products with
sparse vectors where the number of required multiplications equal to
$\overline{K}$ instead of $L$. The HT and the GT with Bridge-$\ell_{0.5}$
shrinkage function, as they where presented in
\ref{sec:performance_adaptPLSTO}, belong to this category with
$\overline{K}=K$ and $\overline{K}=K+P$, respectively. The SCAD
thresholding rule does not guarantee a fixed number of zeros after its
application. This is also the case of the APWL1
\cite{Kopsinis.Slavakis.TheodoridisIEEESP2011}. Moreover, in the case of
the APWL1, exact projections onto the weighted $\ell_1$-ball need to be
computed, and in order to do so, the sorting of a vector is necessary,
which requires in general $\mathcal{O}(L\log_2L)$ operations. However, by
adopting a divide-and-conquer approach, as in \cite{Duchi08}, one might
reduce the above computational complexity down to $O(L)$ operations.

The worst-case computational complexities of all the methods employed are
given in Table~\ref{table:complexities.theory}. The parameter $e_1$ is
either 1 or 2, depending on whether all $\omega_i^{(n)}$ of the APGT are
given the same value or not. In the examples of this paper the former is
the case, i.e., $e_1=1$. Moreover, parameter $e_2$ is either $1$, if the
$\ell_2$ norm of the input vectors $(\bm{u}_n)_{n\in\Natural}$ is
arbitrary, or $0$, if it is normalized to unity.

\begin{table*}[t]
\footnotesize\centering
\begin{tabular}{lcccc}
\toprule
\multicolumn{1}{c}{Methods} & \multicolumn{4}{c}{Operations}\\
\midrule
 & Multiplications & Divisions & Sortings & Powers \\
\midrule
APGT-AT-HT & $(qe_1 + e_2 + 1)L+ (K+e_1+1)q$ & $e_2+1$
& $\mathcal{O}(L)$ & - \\
APGT-AT-$\ell_{0.5}$ & $(qe_1 + e_2 + 1)L+ (K+P+e_1+1)q
+12P+1$ & $P+e_2+2$ & $\mathcal{O}(L)$ & $3P+1$ \\
APGT-AT-SCAD & $(qe_1 + e_2 + 1)L+ (L+e_1+1)q +
  (L-K)$ & $L-K+e_2+1$ & $\mathcal{O}(L)$ & - \\
APWL1 & $(qe_1 + e_2 + 1)L+ (L+e_1+1)q +3L$ &
$2L+e_2+1$ & $\mathcal{O}(L)$ & -  \\
OCCD-TWL & $3L^2+3L$ & $L$ & - & - \\
IPAPA & $\mathcal{O}(q^3)+ (q^2+3q+1)L + q$ & - & - & - \\
\bottomrule
\end{tabular}
\caption{Computational complexities of all the methods
  employed.}\label{table:complexities.theory}
\end{table*}

\section{Conclusions}

The present paper contributed to sparsity-aware online learning tasks in
the following three ways: (i) it established a Generalized Thresholding
(GT) mapping, which can incorporate as a shrinkage function the majority of
the thresholding rules found in the literature, (iii) it proposed a
non-convexly constrained, online learning algorithm for sparse signal
recovery tasks with a computational complexity which scales linearly to the
number of unknowns, and (iii) it introduced a family of mappings which
serves as the wide functional analytic stage for the study of the previous
GT operator. Rigorous discussions on the properties of all the previous
functional analytic tools, as well as a convergence analysis of the
proposed algorithm were provided. To validate the theoretical findings
regarding our algorithm, extensive experiments were conducted, which showed
that the proposed methodology offers a sound theoretical, and very
competitive time-adaptive technique, with lower computational complexity
than several of the state-of-the-art, sparsity-promoting, online learning
algorithms.

\appendix

\section{Convex Sets, Convex Functions, and Projection
  Mappings}\label{sec:convex.analysis}

A subset $C$ of $\Real^L$ will be called convex, if for any $\bm{a},
\bm{a}' \in C$, the line segment $\{\lambda \bm{a} +
(1-\lambda)\bm{a}': \lambda\in [0,1] \}$ lies in $C$. A function
$\Theta:\Real^L \rightarrow\Real$ is called convex if $\forall \bm{a},
\bm{a}'\in \Real^L$, and $\forall \lambda\in[0,1]$, we have
$\Theta\bigl(\lambda\bm{a} + (1-\lambda)\bm{a}'\bigr) \leq
\lambda\Theta(\bm{a}) + (1-\lambda)\Theta(\bm{a}')$. The
\textit{$0$-th level set} of the convex $\Theta$ is defined as
$\lev{0}(\Theta) \coloneqq \bigl\{\bm{a}\in\Real^L: \Theta(\bm{a})
\leq 0\bigr\}$. A subgradient of the convex function $\Theta$ at a
point $\bm{a}$, denoted as $\Theta'(\bm{a})$, is an $L$-dimensional
vector such that $(\bm{v}-\bm{a})^\top \Theta'(\bm{a}) + \Theta(\bm{a})
\leq \Theta(\bm{v})$, $\forall \bm{v}\in \Real^L$. In general, the
number of the subgradients of $\Theta$ at $\bm{a}$ is infinite. The
set of all subgradients of $\Theta$ at a point $\bm{a}$ is called
subdifferential, and it is denoted by $\partial\Theta(\bm{a})$. In the
case where $\Theta$ is differentiable at $\bm{a}$, then the
subgradient $\Theta'(\bm{a})$ is unique, and it is nothing but the
gradient of $\Theta$ at $\bm{a}$.

Given a closed convex $C\subset\Real^L$, define the
\textit{(metric) distance function $d(\cdot,C):
  \Real^L\rightarrow\Real$ to $C$}\/ as follows: $\forall
\bm{a}\in\Real^L$, $d(\bm{a},C)\coloneqq \inf\{\norm{\bm{a}-
  \bm{v}}:\ \bm{v}\in C\}$. Notice that $d(\cdot,C)$ is convex with
$\lev{0}d(\cdot,C)=C$. The \textit{(metric) projection onto $C$}\/ is
defined as the mapping $P_C: \Real^L \rightarrow C$, which maps an
$\bm{a}\in\Real^L$ to the \textit{unique} $P_C(\bm{a})\in C$, such that
$\norm{\bm{a}-P_C(\bm{a})} = d(\bm{a},C)$. For example, the subdifferential of
$d(\cdot,C)$ is given as follows:
\begin{equation}
\partial d(\bm{a},C) = 
\begin{cases}
 N_C(\bm{a}) \cap B[0,1], & \text{if}\ \bm{a}\in C,\\ 
\frac{\bm{a}-P_C(\bm{a})}{d(\bm{a},C)}, & \text{if}\ \bm{a}\notin C, 
\end{cases} \label{subdiff.distance}
\end{equation}
where $N_C(\bm{a})\coloneqq\{\bm{v}\in\Real^L: \bm{v}^\top (\bm{y}-
\bm{a}) \leq 0, \forall \bm{y}\in C\}$.

\section{The Penalized Least-Squares Task}\label{sec:PLSTO}

Going back to \eqref{eq:regression_Model2}, choose $N\in \Naturalstar$, and
define $\bm{U}_n \coloneqq [\bm{u}_n, \bm{u}_{n-1}, \ldots,
  \bm{u}_{n-N+1}]\in \Real^{L\times N}$, as well as $\bm{y}_n \coloneqq
[y_n, y_{n-1}, \ldots, y_{n-N+1}]^{\top}\in\Real^N$, and $\bm{v}_n \coloneqq
[v_n, v_{n-1}, \ldots, v_{n-N+1}]^{\top} \in \Real^N$. Then, it can be easily
verified that \eqref{eq:regression_Model2} takes the form of $\bm{y}_n =
\bm{U}_n^{\top} \bm{a}_* + \bm{v}_n$, $\forall n\in\Natural$. The mainstream of
the batch sparsity-promoting algorithms utilize all the gathered $N$
training data to find an exact or approximate solution, in most cases
iteratively, to the following \textit{penalized least-squares} minimization
task,
\begin{equation}
\min_{\bm{a}\in \Real^L} \frac{1}{2}\norm{\bm{y}_n - \bm{U}_n^{\top}
  \bm{a}}^2 + \lambda \sum_{i=1}^Lp(|a_i|), \label{main.batch.task}
\end{equation}
where $p:\Real\rightarrow [0,\infty)$ stands for a sparsity-promoting and
  non-convex, in general, \textit{penalty} function, $\lambda\in (0,\infty)$ is the
  \textit{regularization} parameter, and $a_i$ stands for the $i$-th
  coordinate of the vector $\bm{a}$.

Choices for $p$ are numerous; if, for example, $p(|a|)
\coloneqq\chi_{\Real\setminus\{0\}}(|a|)$, $\forall a\in \Real$, where
$\chi_{\mathscr{A}}$ stands for the characteristic function with respect to
$\mathscr{A} \subset \Real$, i.e., $\chi_{\mathscr{A}}(\alpha) \coloneqq
1$, if $\alpha\in \mathscr{A}$, and $\chi_{\mathscr{A}}(\alpha) \coloneqq
0$, if $\alpha\notin \mathscr{A}$, then the regularization term
$\sum_{i=1}^Lp(|a_i|)$ becomes the $\ell_0$-norm of $\bm{a}$. In the case
where $p(|a|) \coloneqq |a|$, $\forall a\in\Real$, then the regularization
term is nothing but the $\ell_1$-norm $\norm{a}_1 \coloneqq
\sum_{i=1}^L|a_i|$, and the task \eqref{main.batch.task} becomes the
celebrated LASSO \cite{Tibshirani.Lasso}. However, it has been observed
that if some of the LASSO's regularity conditions are violated, then
LASSO is sub-optimal for model selection \cite{Zou2006, ZouLi2008,
  Zhang2010, Mazumder2011, Frank1993}. Such a behavior has motivated the
search for \textit{non-convex} penalty functions $p$, which bridge the gap
between the $\ell_0$- and $\ell_1$-norm; for example, the $\ell_{\gamma}$
penalty, for $\gamma\in (0,1)$, \cite{Frank1993}, the $\log$
\cite{Antoniadis2001}, the SCAD \cite{Fan_Li2001, Antoniadis2001}, the MC+
\cite{Zhang2010, Mazumder2011}, and the transformed $\ell_1$
\cite{Antoniadis2001} penalties.

Recently, sparsity-promoting coordinate-wise optimization techniques for
solving the task \eqref{main.batch.task} are attracting a lot of interest
\cite{Mazumder2011, SimonEtal2011, Angelosantejournal2010}. To be more
concrete, assume, for example, that $N=L$, and that the matrix $\bm{U}_n$ is
orthogonal. Byy defining $\tilde{\bm{a}}_n \coloneqq \bm{U}_n\bm{y}_n$,
\eqref{main.batch.task} can be equivalently viewed as the following
separable optimization task \cite{Antoniadis2001, Antoniadis2007},
\begin{equation}
\min_{\bm{a}\in \Real^L} \sum_{i=1}^L \frac{1}{2\lambda}
\left(\tilde{a}_i - a_i\right)^2 +
p(|a_i|).\label{coordinate.batch.task}
\end{equation}
Under some mild regularity conditions on $p$ \cite{Antoniadis2001}, the
minimization task of \eqref{coordinate.batch.task} possesses a unique
minimizer. Due to the separability of \eqref{coordinate.batch.task} in
coordinates, the minimization task of \eqref{coordinate.batch.task} can be
viewed as a task defined on an $1$-dimensional axis, instead of an
$L$-dimensional domain. Accordingly, the problem reduces to the univariate
PLS task described in \eqref{def.upls}.

Figs.~\ref{fig:variousthresholds1}(b-d), show the thresholding functions
(PLSTO, see \eqref{def.plsto}), which solve \eqref{def.upls} for some of
the most commonly employed penalty functions. For example, if $p(|a|)
\coloneqq \left[\lambda^2 - (|a|-\lambda)^2
  \chi_{[0,\lambda)}(|a|)\right]/\lambda$, $\forall a\in\Real$, then the
  resulting PLSTO is the celebrated \textit{Hard Thresholding (HT)} mapping
  \cite{Antoniadis2001}, which is depicted in
  Fig.~\ref{fig:variousthresholds1}a together with the well-known
  \textit{Soft Thresholding (ST)} mapping which results in the case where
  $p(|a|) \coloneqq |a|$, i.e. is chosen such that to lead to the LASSO
  task. Note that both ST and HT operators have been effectively employed
  in iterative thresholding schemes for fast sparse signal recovery under
  the compressed sensing framework \cite{Daubechies.Fornasier.Loris2008,
    BlumensathJST2010, Blumensath.IHT, Foucart.IHT}. The rest of the
  thresholding rules, shown in Fig.~\ref{fig:variousthresholds1}b
  correspond to the MC+ penalty \cite{Zhang2010, Mazumder2011} and the SCAD
  \cite{Fan_Li2001}, respectively. Both SCAD and MC+ leave large components
  unchanged, like HT, while avoiding being discontinuous and at the same
  time allowing a linear/gradual transition between the ``kill'' and the
  ``keep'' areas of HT. HT is far from being the only discontinuous
  thresholding operator. An example is shown in
  Figs.~\ref{fig:variousthresholds1}c, by the widely known Bridge threshold
  \cite{Frank1993}, which is related to the $\ell_{\gamma}$ penalty,
  $\gamma<1$ \cite{LorenzTikhonovRegularization}. Note that this
  thresholding rule comprise nonlinear segments. Continuous thresholding
  functions, that contain nonlinear parts, are shown in
  Fig.~\ref{fig:variousthresholds1}(d). More specifically, the non-negative
  garrote \cite{Gao1998} and representatives of the n-degree garrote
  threshold are shown. Similar thresholding functions are also the
  hyperbolic shrinkage rule \cite{Tao2000} and PLSTO's stemming from the
  nonlinear diffusive filtering approach \cite{Antoniadis2007}.

\section{Properties of the GT Mapping} \label{sec:properties.T_K}

\begin{theorem}\label{thm:properties.T_K}\mbox{}
\begin{enumerate}[leftmargin=0pt,itemindent=20pt]
\item\label{thm:same.support} $\forall\bm{x}\in \Real^L$,
  $J_{T_{\text{GT}}^{(K)}(\bm{x})}^{(K)}= J_{\bm{x}}^{(K)}$.
\item\label{thm:fix.T_K} $\Fix(T_{\text{GT}}^{(K)})= \bigcup_{J\in
  \mathscr{T}(K,L)} M_J$. Notice, here, that $\Fix(T_{\text{GT}}^{(K)})$,
  as a union of subspaces, is non-convex.
\item\label{thm:demiclosed} Let a sequence $(\bm{x}_n)_{n\in
  \Natural}\subset \Real^L$ and an $\bm{x}_*\in \Real^L$. If
  $\lim_{n\rightarrow\infty} \bm{x}_n = \bm{x}_*$, and
  $\lim_{n\rightarrow\infty} \bigl(I-T_{\text{GT}}^{(K)}\bigr)(\bm{x}_n) =
  \bm{0}$, then $\bm{x}_*\in \Fix(T_{\text{GT}}^{(K)})$. This property can
  be rephrased as $I-T_{\text{GT}}^{(K)}$ being \textit{demiclosed at
    $\bm{0}$} \cite{GoebelKirk}.
\item\label{thm:T_K.pqne} $T_{\text{GT}}^{(K)}$ is \textit{$1$-attracting
  partially quasi-nonexpansive,} i.e., $\forall \bm{x}\in \Real^L$,
  $\forall \bm{y}\in
  M_{J_{\bm{x}}^{(K)}}, \norm{\bm{x}-T_{\text{GT}}^{(K)}(\bm{x})}^2 \leq
  \norm{\bm{x}-\bm{y}}^2 - \norm{T_{\text{GT}}^{(K)}(\bm{x}) - \bm{y}}^2$.
\end{enumerate}
\end{theorem}

\noindent\textit{Proof:}
\begin{enumerate}[leftmargin=0pt,itemindent=20pt]
\item Define $\bm{z} \coloneqq T_{\text{GT}}^{(K)}(\bm{x})$. In
order to derive a contradiction, assume that $J_{\bm{x}}^{(K)}\neq
J_{\bm{z}}^{(K)}$. Since both $J_{\bm{x}}^{(K)}, J_{\bm{z}}^{(K)}$ have the
same cardinality, the previous assumption means that there exist $l_0,
l_0'$ such that $l_0\in J_{\bm{x}}^{(K)} \setminus J_{\bm{z}}^{(K)}$, and
$l_0'\in J_{\bm{z}}^{(K)} \setminus J_{\bm{x}}^{(K)}$. Hence, $|x_{l_0'}| =
|\shr(x_{l_0'})| = |z_{l_0'}|\geq |z_{l_0}| = |x_{l_0}| \geq
|x_{l_0'}|$. The previous result implies that $|x_{l_0}| = |x_{l_0'}|$,
which, in turn, suggests by the definition of $J_{\bm{x}}^{(K)}$ that
$l_0<l_0'$. Moreover, $|z_{l_0'}| = |z_{l_0}|$ and $l_0'<l_0$ by the
definition of $J_{\bm{z}}^{(K)}$. Thus, $l_0<l_0'<l_0$, which is absurd.
This contradiction establishes the claim of
Thm.~\ref{thm:properties.T_K}.\ref{thm:same.support}.

\item Pick any $\bm{x}\in \bigcup_{J\in \mathscr{T}(K,L)} M_J$. It
is easy to verify by Def.~\ref{def:GT} that
$T_{\text{GT}}^{(K)}(\bm{x})=\bm{x}$, i.e., $\bm{x}\in
\Fix(T_{\text{GT}}^{(K)})$. To prove the opposite inclusion, assume any
$\bm{x}\in \Fix(T_{\text{GT}}^{(K)})$, i.e.,
$T_{\text{GT}}^{(K)}(\bm{x})=\bm{x}$. Since $\forall l\in
J_{\bm{x}}^{(K)}$, the relation $T_{\text{GT}}^{(K)}(\bm{x})=\bm{x}$ leads
to the trivial result $x_l= x_l$, we deal here only with the more
interesting case of $l\notin J_{\bm{x}}^{(K)}$. For such an $l$, according
to Def.~\ref{def:GT}, we must have $\shr(x_l)=x_l$, which implies that
$|\shr(x_l)|=|x_l|$. However, by the properties of $\shr$, given in
Defs.~\ref{def:GT}.\ref{shrinks} and \ref{def:GT}.\ref{strictly.shrinks},
we necessarily obtain that $x_l=0$. Since this holds $\forall l\notin
J_{\bm{x}}^{(K)}$, Def.~\ref{def:special.subspace} suggests that $\bm{x}\in
M_{J_{\bm{x}}^{(K)}}$. Now, recall that $J_{\bm{x}}^{(K)}\in
\mathscr{T}(K,L)$ to establish the inclusion $\bm{x}\in \bigcup_{J\in
  \mathscr{T}(K,L)} M_J$.

\item

\begin{enumerate}[leftmargin=0pt,itemindent=20pt]

\item Assume, for a contradiction, that there exists an
$\varepsilon>0$ and a subsequence $(n_k)_{k\in\Natural}$, such that
$\left|x_{n_k,l_{n_k}} \right|\geq \varepsilon$, $\forall l_{n_k} \notin
J_{\bm{x}_{n_k}}^{(K)}$, $\forall k\in\Natural$.

By Def.~\ref{def:GT}.\ref{strictly.shrinks}, $\exists \delta>0$ such that
$\bigl|\shr(x_{n_k,l_{n_k}}) \bigr|\leq
\bigl|x_{n_k,l_{n_k}}\bigr|-\delta$, $\forall k$. Then, it is easy to
verify that $\forall k$, $\bigl|x_{n_k,l_{n_k}} -
\shr(x_{n_k,l_{n_k}})\bigr| \geq \bigl|x_{n_k,l_{n_k}} \bigr| -
\bigl|\shr(x_{n_k,l_{n_k}})\bigr| \geq \bigl|x_{n_k,l_{n_k}}\bigr| -
\bigl|x_{n_k,l_{n_k}} \bigr| + \delta = \delta$. This implies that $\forall k$,
\begin{equation}
\sum_{l\notin J_{\bm{x}_{n_k}}^{(K)}}\bigl(x_{n_k,l}-
\shr(x_{n_k,l})\bigr)^2 \geq \bigl(x_{n_k,l_{n_k}}- \shr(x_{n_k,l_{n_k}})\bigr)^2
\geq \delta^2.\label{to.derive.a.contrad}
\end{equation}

Notice that $\norm{\bigl(I-T_{\text{GT}}^{(K)}\bigr)(\bm{x}_n)}^2=
\sum_{l\notin J_{\bm{x}_n}^{(K)}}\bigl(x_{n,l}-
\shr(x_{n,l})\bigr)^2$. Hence, the assumption that
$\lim_{n\rightarrow\infty}\bigl(I-T_{\text{GT}}^{(K)} \bigr)(\bm{x}_n) =
\bm{0}$ implies that for the $\delta$ of \eqref{to.derive.a.contrad},
$\exists n_0\in \Natural$ such that $\forall n\geq n_0$, $\sum_{l\notin
  J_{\bm{x}_n}^{(K)}}\bigl(x_{n,l}- \shr(x_{n,l})\bigr)^2 < \delta^2$. This
contradicts \eqref{to.derive.a.contrad}. In other words, our initial claim
is wrong, and the contrary proposition becomes: $\forall \varepsilon>0$,
there exists an $n_0\in\Natural$ such that $|x_{n,l}| <\varepsilon$,
$\forall l\notin J_{\bm{x}_n}^{(K)}$, $\forall n\geq n_0$. This can be
equivalently written in a more compact form as follows:
\begin{equation}
\lim_{n\rightarrow\infty} \max\bigl\{|x_{n,l}|: l\notin
J_{\bm{x}_n}^{(K)}\bigr\}=0. \label{a.good.property}
\end{equation}

\item Let us define here
\begin{equation}
J_{\infty} \coloneqq \liminf_{n\rightarrow\infty} J_{\bm{x}_n}^{(K)}
\coloneqq \bigcup_{n=0}^{\infty} \bigcap_{m=n}^{\infty}
J_{\bm{x}_m}^{(K)}. \label{liminf.J}
\end{equation}
In words, $J_{\infty}$ contains all those points which belong to
all but a finite number of $J_{\bm{x}_n}^{(K)}$s. There are two cases
regarding $J_{\bm{x}_*}^{(K)}$ and $J_{\infty}$; either
$J_{\bm{x}_*}^{(K)} \cap J_{\infty} \neq \emptyset$ or
$J_{\bm{x}_*}^{(K)} \cap J_{\infty} = \emptyset$. Notice that
the latter covers also the case where $J_{\infty}=
\emptyset$. Let us examine each case separately.

\begin{enumerate}[leftmargin=0pt,itemindent=20pt]

\item The case of $J_{\bm{x}_*}^{(K)} \cap J_{\infty} \neq
  \emptyset$.

\begin{enumerate}[leftmargin=0pt,itemindent=20pt]

\item Assume that $J_{\bm{x}_*}^{(K)} \subset
  J_{\infty}$. This implies that there exists an $n_0$ such that
  $J_{\bm{x}_*}^{(K)} \subset \bigcap_{n\geq n_0}
  J_{\bm{x}_n}^{(K)}$. Since both $J_{\bm{x}_*}^{(K)}$ and
  $J_{\bm{x}_n}^{(K)}$ have the same cardinality, i.e., $K$, we obtain
  that $\forall n\geq n_0$, $J_{\bm{x}_*}^{(K)}=
  J_{\bm{x}_n}^{(K)}$. Choose any $l\notin J_{\bm{x}_*}^{(K)}=
  J_{\bm{x}_n}^{(K)}$. By \eqref{a.good.property},
  $\lim_{n\rightarrow\infty} x_{n,l}= 0 = x_{*,l}$. Thus, $\forall
  l\notin J_{\bm{x}_*}^{(K)}$, $x_{*,l}=0$, or equivalently,
  $\bm{x}_*\in \bigcup_{J\in\mathscr{T}(K,L)} M_J$.

\item Assume now that $J_{\bm{x}_*}^{(K)} \not\subset
  J_{\infty}$. Hence, there exists an $l\in J_{\bm{x}_*}^{(K)}$
  and a subsequence $(n_k)_{k\in\Natural}$ such that $l\notin
  J_{\bm{x}_{n_k}}^{(K)}$, $\forall k\in\Natural$. By
  \eqref{a.good.property}, $\lim_{n\rightarrow\infty} x_{n,l} = 0 =
  x_{*,l}$. Since $l\in J_{\bm{x}_*}^{(K)}$, we clearly have that
  $x_{*,l'}=0$, $\forall l'\notin J_{\bm{x}_*}^{(K)}$. Hence,
  $\bm{x}_*\in \bigcup_{J\in\mathscr{T}(K,L)} M_J$.

\end{enumerate}

\item The case of $J_{\bm{x}_*}^{(K)} \cap J_{\infty} =
  \emptyset$. This means that there exists an $l\in
  J_{\bm{x}_*}^{(K)}$ and a subsequence $(n_k)_{k\in\Natural}$ such
  that $l\notin J_{\bm{x}_{n_k}}^{(K)}$, $\forall k\in\Natural$. Thus,
  similarly to our previous arguments, $\bm{x}_*\in
  \bigcup_{J\in\mathscr{T}(K,L)} M_J$.
\end{enumerate}
\end{enumerate}

\item Define $R_K \coloneqq 2T_{\text{GT}}^{(K)}-I$. Given any
  $\bm{x}\in \Real^L$, let $\bm{z} \coloneqq
  T_{\text{GT}}^{(K)}(\bm{x})$, as in Def.~\ref{def:GT}. Then, verify
  that $\forall \bm{y}\in M_{J_{\bm{x}}^{(K)}}$,
\begin{align*}
\norm{R_K(\bm{x}) - \bm{y}}^2 & = \sum_{l=1}^L (2z_l - x_l -
y_l)^2\\
& = \sum_{l\in J_{\bm{x}}^{(K)}} (x_l - y_l)^2 +
\sum_{l\notin J_{\bm{x}}^{(K)}} (2\shr(x_l) - x_l)^2\\
& \leq \sum_{l\in J_{\bm{x}}^{(K)}} (x_l - y_l)^2 +
\sum_{l\notin J_{\bm{x}}^{(K)}} x_l^2 =
\norm{\bm{x}-\bm{y}}^2.
\end{align*}
The previous inequality is obtained from the observation that the
properties of $\shr$ in Def.~\ref{def:GT} suggest $\shr^2(x_l) \leq
x_l\shr(x_l)$, and from the following elementary calculations:
$(2\shr(x_l)-x_l)^2= 4\shr^2(x_l) + x_l^2 - 4x_l\shr(x_l)\leq 4\shr^2(x_l)
+ x_l^2 - 4\shr^2(x_l)$. Hence, $\forall \bm{x}\in\Real^L$,
$\forall\bm{y}\in M_{J_{\bm{x}}^{(K)}}$, $\norm{R_K(\bm{x})-\bm{y}}^2 \leq
\norm{\bm{x}-\bm{y}}^2 \Leftrightarrow
\norm{2T_{\text{GT}}^{(K)}(\bm{x})-\bm{x} -\bm{y}}^2 \leq
\norm{\bm{x}-\bm{y}}^2 \Leftrightarrow
\norm{2\bigl(T_{\text{GT}}^{(K)}(\bm{x})-\bm{y}\bigr) - (\bm{x} -\bm{y})}^2
\leq \norm{\bm{x}-\bm{y}}^2 \Leftrightarrow
\norm{\bm{x}-T_{\text{GT}}^{(K)}(\bm{x})}^2 \leq \norm{\bm{x} - \bm{y}}^2 -
\norm{T_{\text{GT}}^{(K)}(\bm{x}) - \bm{y}}^2$, where in order to obtain
the last equivalence we used some elementary algebra, and the
fact\\ $2\innprod{\bm{x}-\bm{y}}{T_{\text{GT}}^{(K)}(\bm{x})-\bm{y}} =
\norm{\bm{x}-\bm{y}}^2 + \norm{T_{\text{GT}}^{(K)}(\bm{x})-\bm{y}}^2 -
\norm{\bm{x}-T_{\text{GT}}^{(K)}(\bm{x})}^2$. This establishes the claim of
Thm.~\ref{thm:properties.T_K}.\ref{thm:T_K.pqne}.
\end{enumerate}

\section{Proof of Theorem~\ref{thm:algo}}\label{sec:analysis.algo}

Let us define first a sequence of convex functions
  $(\Theta_n)_{n\in\Natural}$ in an inductive way. Given the
  time index $n$, and the estimate $\bm{a}_n\in \Real^L$, define the
  following convex function; $\forall \bm{a}\in \Real^L$,
\begin{equation*}
\Theta_n(\bm{a}) \coloneqq \begin{cases}
\sum_{i\in\mathcal{I}_n} \frac{\omega_i^{(n)} d(\bm{a}_n,
  S_i[\epsilon_i])} {\sum_{j\in\mathcal{I}_n}\omega_j^{(n)} d(\bm{a}_n,
  S_j[\epsilon_j])} d(\bm{a}, S_i[\epsilon_i]) & \\
\quad = \frac{1}{L_n}\sum_{i\in\mathcal{I}_n} \omega_i^{(n)} d(\bm{a}_n,
  S_i[\epsilon_i]) d(\bm{a}, S_i[\epsilon_i]), \\
& \hspace{-50pt}\text{if}\ \mathcal{I}_n
\neq \emptyset,\\
0, & \hspace{-50pt} \text{otherwise},
\end{cases}
\end{equation*}
where $L_n \coloneqq \sum_{j\in\mathcal{I}_n}\omega_j^{(n)}
d(\bm{a}_n, S_j[\epsilon_j])$. It is easy to verify by the definition of
$\mathcal{I}_n$, that if $\mathcal{I}_n\neq \emptyset$, then $\forall
i\in \mathcal{I}_n$, $d(\bm{a}_n, S_i[\epsilon_i])>0$, and thus $L_n>
0$. Moreover, if $\mathcal{I}_n= \emptyset$, then
$\Theta_n'(\bm{a})=0$, $\forall \bm{a}$.

Let us look closer to $\Theta_n$, and especially only the interesting
case of $\mathcal{I}_n\neq \emptyset$. By standard subgradient
calculus, it can be verified by \eqref{subdiff.distance} that
\begin{equation*}
\Theta_n'(\bm{a}_n) = \frac{1}{L_n} \sum_{i\in\mathcal{I}_n}
\omega_i^{(n)} d(\bm{a}_n, S_i[\epsilon_i]) \frac{\bm{a}_n -
  P_{S_i[\epsilon_i]}(\bm{a}_n)}{d(\bm{a}_n,S_i[\epsilon_i])} =
\frac{1}{L_n} \sum_{i\in\mathcal{I}_n} \omega_i^{(n)} \bigl(\bm{a}_n -
P_{S_i[\epsilon_i]}(\bm{a}_n)\bigr).
\end{equation*}
Thus, whenever $\mathcal{I}_n\neq\emptyset$, we have
$\Theta_n'(\bm{a}_n)= \bm{0}$ iff $\sum_{i\in\mathcal{I}_n}
\omega_i^{(n)} \bigl(\bm{a}_n - P_{S_i[\epsilon_i]}(\bm{a}_n)\bigr)=
\bm{0}$. Hence, for some user-defined parameter $\lambda_n>0$, it
is straightforward to see that
\begin{equation*}
\bm{a}_n - \lambda_n
\frac{\Theta_n(\bm{a}_n)} {\norm{\Theta_n'(\bm{a}_n)}^2}
\Theta_n'(\bm{a}_n) = \bm{a}_n - \lambda_n
\frac{\sum_{i\in\mathcal{I}_n} \omega_i^{(n)} d^2(\bm{a}_n,
  S_i[\epsilon_i])}{\norm{\sum_{i\in\mathcal{I}_n} \omega_i^{(n)}
  \bigl(\bm{a}_n - P_{S_i[\epsilon_i]}(\bm{a}_n)\bigr)}^2}
\sum_{i\in\mathcal{I}_n} \omega_i^{(n)} \bigl(\bm{a}_n -
P_{S_i[\epsilon_i]}(\bm{a}_n)\bigr).
\end{equation*}
If we let $\lambda_n \coloneqq \mu_n/\mathcal{M}_n$, then an
examination of \eqref{algo}, for both the cases of $\mathcal{I}_n\neq
\emptyset$ and $\mathcal{I}_n = \emptyset$, implies that the proposed
algorithm can be rephrased as follows; for $\lambda_n \coloneqq
\mu_n/\mathcal{M}_n \in [\varepsilon', 2-\varepsilon']$,
\begin{equation}
\bm{a}_{n+1} = \begin{cases}
T_{\text{GT}}^{(K)} \Bigl(\bm{a}_n - \lambda_n
\frac{\Theta_n(\bm{a}_n)} {\norm{\Theta_n'(\bm{a}_n)}^2}
\Theta_n'(\bm{a}_n) \Bigr), & \Theta_n'(\bm{a}_n)\neq \bm{0},\\
T_{\text{GT}}^{(K)}(\bm{a}_n), & \Theta_n'(\bm{a}_n)= \bm{0}.
\end{cases}\label{algo.b}
\end{equation}

\begin{enumerate}[leftmargin=0pt,itemindent=20pt]

\item Fix any $\bm{v}\in \Omega_n$, and assume that
  $\Theta_n'(\bm{a}_n) \neq 0$. Then, since $\bm{v}\in
  M_{J_{\bm{a}_n}^{(K)}} \subset \Fix(T_{\text{GT}}^{(K)})$,
\begin{align}
\norm{\bm{a}_{n+1} - \bm{v}}^2 & = \norm{ T_{\text{GT}}^{(K)}
  \Bigl(\bm{a}_n - \lambda_n \frac{\Theta_n(\bm{a}_n)}
       {\norm{\Theta_n'(\bm{a}_n)}^2} \Theta_n'(\bm{a}_n)\Bigr) -
       \bm{v}}^2 \nonumber \\ &\leq \norm{ (\bm{a}_n - \bm{v}) - \lambda_n
  \frac{\Theta_n(\bm{a}_n)} {\norm{\Theta_n'(\bm{a}_n)}^2}
  \Theta_n'(\bm{a}_n)}^2 \nonumber \\ & \leq \norm{\bm{a}_n - \bm{v}}^2
- \lambda_n(2-\lambda_n)
\frac{\Theta_n^2(\bm{a}_n)}{\norm{\Theta_n'(\bm{a}_n)}^2} \nonumber\\
& \leq \norm{\bm{a}_n - \bm{v}}^2 - (\varepsilon')^2
\frac{\Theta_n^2(\bm{a}_n)}{\norm{\Theta_n'(\bm{a}_n)}^2}, \label{1st.inequality}
\end{align}
where we used the property of
Thm.~\ref{thm:properties.T_K}.\ref{thm:T_K.pqne}, and the definition
of the subgradient $\Theta_n'(\bm{a}_n)$. As a result, $\norm{\bm{a}_{n+1}
  - \bm{v}} \leq \norm{\bm{a}_n - \bm{v}}$. Notice, that this holds
true also for the case where $\Theta_n'(\bm{a}_n) = 0$. Now, if we
apply $\inf_{\bm{v}\in \Omega_n}$ on both sides of the previous
inequality, then we establish the claim of
Thm.~\ref{thm:algo}.\ref{simple.monotinicity}.

\item Fix $n\in \overline{n_0, n_0+N-1}$. Assume that
  $\Theta_n'(\bm{a}_n) \neq 0$. Then, notice by the convexity of the
  function $\norm{\cdot}^2$ that
  \begin{align}
    \frac{\Theta_n^2(\bm{a}_n)}{\norm{\Theta_n'(\bm{a}_n)}^2} & =
    \frac{\sum\limits_{i\in\mathcal{I}_n} \omega_i^{(n)} d^2(\bm{a}_n,
      S_i[\epsilon_i]) \sum\limits_{i\in\mathcal{I}_n} \omega_i^{(n)} d^2(\bm{a}_n,
    S_i[\epsilon_i])}{\norm{\sum_{i\in\mathcal{I}_n} \omega_i^{(n)}
        \bigl(\bm{a}_n - P_{S_i[\epsilon_i]}(\bm{a}_n)\bigr)}^2}\nonumber\\
     & \geq \frac{1}{q} \sum_{i\in\mathcal{I}_n}
    d^2(\bm{a}_n, S_i[\epsilon_i]) \nonumber\\ &\geq \frac{1}{q} \max
  \bigl\{d^2(\bm{a}_n, S_j[\epsilon_j]):
  j\in\mathcal{J}_n\bigr\}. \label{fraction.Theta}
  \end{align}
Hence, by \eqref{fraction.Theta},
\begin{align}
\norm{\bm{a}_{n+1} - \bm{v}}^2 & \leq \norm{\bm{a}_n - \bm{v}}^2
\nonumber\\ &\hspace{10pt} - \frac{(\varepsilon')^2}{q}
\max\bigl\{d^2(\bm{a}_n,S_j[\epsilon_j]):
j\in\mathcal{J}_n\bigr\}. \label{2nd.inequality}
\end{align}
Notice also that \eqref{2nd.inequality} holds true also for the case where
$\Theta_n'(\bm{a}_n) = 0$. If we take the infimum over all $\bm{v}\in
\bigcap_{n=n_0}^{n_0+N-1} \Omega_n$ on both sides of
\eqref{2nd.inequality}, and if we add the resulting inequality for all
values of $n\in \overline{n_0,n_0+N-1}$, then the claim of
Thm.~\ref{thm:algo}.\ref{bound} is established.

\item

\begin{enumerate}[leftmargin=0pt,itemindent=20pt]

\item Choose arbitrarily any $\bm{v}\in \Omega$. Then, by
definition, there exists an $n_0$ such that $\bm{v}\in \bigcap_{n\geq
  n_0}\Omega_n$. Clearly, \eqref{2nd.inequality} leads to
$\norm{\bm{a}_{n+1} - \bm{v}} \leq \norm{\bm{a}_n - \bm{v}}$, $\forall
n\geq n_0$, i.e., $(\norm{\bm{a}_n- \bm{v}})_{n\geq n_0}$ is monotonically
non-increasing, and thus convergent. This result implies also that the
sequence $(\bm{a}_n)_{n\in\Natural}$ is bounded, and that
$\mathfrak{C}\bigl((\bm{a}_n)_{n\in\Natural} \bigr) \neq \emptyset$
\cite{bauschke.combettes.book}.

\item Let $\bm{v}\in \Omega$ and $n_0$ as previously. We have already
  seen that $\forall\bm{v}\in \Omega$, $(\norm{\bm{a}_n -
    \bm{v}}^2)_{n\in\Natural}$ is convergent. Thus, it is Cauchy, and
\begin{equation}
\lim_{n\rightarrow\infty} \bigl(\norm{\bm{a}_n -
    \bm{v}}^2 - \norm{\bm{a}_{n+1} - \bm{v}}^2 \bigr) = 0. \label{cauchy}
\end{equation}

Now, a simple inspection of \eqref{2nd.inequality} and \eqref{cauchy}
establish the claim of
Thm.~\ref{thm:algo}.\ref{thm:distance.goes.2.zero}.

\item Notice that $\forall n\geq n_0$, $\Theta_n'(\bm{a}_n)= \bm{0}$
  iff $\bm{a}_n\in \lev{0}(\Theta_n)= \bigcap_{i\in\mathcal{I}_n}
  S_i[\epsilon_i]\neq \emptyset$. Hence, for such $n$, \eqref{algo.b}
  takes the following equivalent form:
\begin{equation}
\bm{a}_{n+1} = T_{\text{GT}}^{(K)}
T_{\Theta_n}^{(\lambda_n)}(\bm{a}_n), \label{algo.c}
\end{equation}
where the mapping $T_{\Theta_n}$ is the \textit{subgradient projection
  mapping with respect to the convex $\Theta_n$} defined as
\cite{bauschke.combettes.book}: $T_{\Theta_n}^{(\lambda_n)}(\bm{a})
\coloneqq \bm{a} - \lambda_n \frac{\Theta_n(\bm{a})}
          {\norm{\Theta_n'(\bm{a})}^2} \Theta_n'(\bm{a})$, if $\bm{a}\notin
          \lev{0}(\Theta_n)$, and $T_{\Theta_n}^{(\lambda_n)}(\bm{a})
          \coloneqq \bm{a}$, if $\bm{a}\in \lev{0}(\Theta_n)$. The
          equations \eqref{1st.inequality} and \eqref{cauchy} imply that
          $\lim_{n\rightarrow\infty}
          \frac{\Theta_n^2(\bm{a}_n)}{\norm{\Theta_n'(\bm{a}_n)}^2} =
          0$. It is a matter of simple algebra to show also that
          $\norm{\bigl( I-T_{\Theta_n}^{(\lambda_n)} \bigr)(\bm{a}_n)} =
          \lambda_n \frac{\Theta_n(\bm{a}_n)}{\norm{\Theta_n'(\bm{a}_n)}}
          \leq 2 \frac{\Theta_n(\bm{a}_n)}{\norm{\Theta_n'(\bm{a}_n)}}$. As
          such,
\begin{equation}
\lim_{n\rightarrow\infty} \bigl( I-T_{\Theta_n}^{(\lambda_n)}
\bigr)(\bm{a}_n) = \bm{0}. \label{I-Tsub.2.zero}
\end{equation}
A remarkable property of the subgradient projection mapping is
  the following \cite{bauschke.combettes.book}:
  $\forall\bm{a}\in\Real^L$, $\forall \bm{v}\in \lev{0}(\Theta_n)\neq
  \emptyset$,
\begin{equation}
\frac{2-\lambda_n}{\lambda_n} \norm{\bm{a}-
  T_{\Theta_n}^{(\lambda_n)}(\bm{a})}^2 \leq \norm{\bm{a}-\bm{v}}^2 -
\norm{T_{\Theta_n}^{(\lambda_n)}(\bm{a}) -
  \bm{v}}^2. \label{subgrad.proj.attracts}
\end{equation}

By \eqref{algo.c} and
Thm.~\ref{thm:properties.T_K}.\ref{thm:same.support}, we can see that
$J^{(K)}_{\bm{a}_{n+1}}=
J^{(K)}_{T_{\Theta_n}^{(\lambda_n)}(\bm{a}_n)}$. Now, notice by
Thm.~\ref{thm:properties.T_K}.\ref{thm:T_K.pqne} and
\eqref{subgrad.proj.attracts} that $\forall\bm{v}\in \Omega$,
\begin{align}
\norm{T_{\Theta_n}^{(\lambda_n)} (\bm{a}_n)- T_{\text{GT}}^{(K)}
  T_{\Theta_n}^{(\lambda_n)} (\bm{a}_n)}^2 & \nonumber\\ & \hspace{-100pt}
\leq \norm{T_{\Theta_n}^{(\lambda_n)} (\bm{a}_n) - \bm{v}}^2 -
\norm{T_{\text{GT}}^{(K)}T_{\Theta_n}^{(\lambda_n)} (\bm{a}_n) - \bm{v}}^2
\nonumber\\ & \hspace{-100pt} = \norm{T_{\Theta_n}^{(\lambda_n)} (\bm{a}_n)
  - \bm{v}}^2 - \norm{\bm{a}_{n+1} - \bm{v}}^2
\nonumber\\ & \hspace{-100pt} \leq \norm{\bm{a}_n - \bm{v}}^2 -
\frac{2-\lambda_n}{\lambda_n} \norm{\bm{a}_n - T_{\Theta_n}^{(\lambda_n)}
  (\bm{a}_n)}^2 \nonumber\\ &\hspace{-90pt} - \norm{\bm{a}_{n+1} - \bm{v}}^2
\leq \norm{\bm{a}_n - \bm{v}}^2 - \norm{\bm{a}_{n+1} - \bm{v}}^2.
\end{align}
Thus, by \eqref{cauchy}, we obtain
\begin{equation}
\lim_{n\rightarrow\infty} \bigl(I-T_{\text{GT}}^{(K)}
  \bigr)T_{\Theta_n}^{(\lambda_n)} (\bm{a}_n) = \bm{0}. \label{I-Tgt.2.0}
\end{equation}

By Thm.~\ref{thm:algo}.\ref{thm:exist.cluster}, choose
any $\bm{\hat{a}}_*\in \mathfrak{C}\bigl((\bm{a}_n)_{n\in\Natural}
\bigr)$. Thus, there exists a subsequence
$(\bm{a}_{n_k})_{k\in\Natural}$ such that $\lim_{k\rightarrow\infty}
\bm{a}_{n_k}= \hat{\bm{a}}_*$. Hence, by \eqref{I-Tsub.2.zero},
$\lim_{k\rightarrow\infty} T_{\Theta_{n_k}}^{(\lambda_{n_k})}
(\bm{a}_{n_k}) = \hat{\bm{a}}_*$. This result, \eqref{I-Tgt.2.0}, and
Thm.~\ref{thm:properties.T_K}.\ref{thm:demiclosed} lead to
$\hat{\bm{a}}_*\in \Fix\bigl(T_{\text{GT}}^{(K)}\bigr)$. Since
$\hat{\bm{a}}_*$ was chosen arbitrarily, we obtain the desired
$\mathfrak{C}\bigl((\bm{a}_n)_{n\in\Natural} \bigr) \subset
\Fix\bigl(T_{\text{GT}}^{(K)}\bigr)$.

\end{enumerate}
\end{enumerate}


\end{document}